\documentclass[11pt,a4paper]{article}

\usepackage[a4paper,
            textwidth=145mm,
            top=28mm, bottom=32mm,
            footskip=14mm,
            headsep=8mm]{geometry}

\usepackage[T1]{fontenc}
\usepackage{mathtools}  
\usepackage{amsthm}
\usepackage{amssymb}

\usepackage{newtxtext}
\usepackage[varg,bigdelims]{newtxmath}
\usepackage{microtype}

\usepackage{stmaryrd}        
\usepackage{mathrsfs}        
\usepackage{stackengine}     
\usepackage{verbatim}

\usepackage{graphicx}
\usepackage{xcolor}
\usepackage{float}
\usepackage{subcaption}
\usepackage{tikz}
\usepackage{tikz-cd}
\usetikzlibrary{fit,calc,petri,arrows.meta,backgrounds,positioning,decorations.pathmorphing}

\usepackage{quiver}          

\colorlet{placecol}{blue!20}
\tikzset{
  proplace/.style    = {place, fill=blue!20, minimum size=5mm, inner sep=0pt},
  trans/.style       = {transition, scale=0.6, fill=black!20},
  ptoken/.style      = {token, scale=1, text=white, font=\tiny,
                        inner sep=0pt, minimum size=2.6mm},
  ptokenblank/.style = {token, scale=1},
  procbox/.style     = {rounded corners=2.5pt, fill=black!7, draw=none, inner sep=0pt},
  mor/.style         = {-{Stealth[length=2mm,width=1.7mm]}, thick, red!80!black,
                        shorten >=1.5pt, shorten <=1.5pt},
}

\def\procOpenKey{open}
\def\procNoTok{}

\newcommand{\procrow}[5]{%
  \def\procType{#5}%
  \def\procTok{#4}%
  \ifx\procTok\procNoTok
    \def\procTokStyle{ptokenblank}%
  \else
    \def\procTokStyle{ptoken}%
  \fi
  \ifx\procType\procOpenKey
    \node[proplace] (#1p#2) at (0.48, #3) {}
      [children are tokens] child {node[\procTokStyle] {#4}};
  \else
    \node[proplace] (#1p#2) at (0, #3) {}
      [children are tokens] child {node[\procTokStyle] {#4}};
    \node[trans] (#1t#2) at (1, #3) {};
    \draw[-latex, thick] (#1p#2) -- (#1t#2);
  \fi
}

\newcommand{\proc}[6]{%
  \begin{scope}[shift={#2}]
    \procrow{#1}{a}{0.45}{#3}{#5}
    \procrow{#1}{b}{-0.25}{#4}{#6}
    \coordinate (#1nw) at (-0.36, 0.82);
    \coordinate (#1se) at ( 1.32,-0.59);
  \end{scope}
  \begin{pgfonlayer}{background}
    \node[procbox, fit=(#1nw)(#1se)] (#1box) {};
  \end{pgfonlayer}
}

\usepackage[textsize=footnotesize]{todonotes}

\usepackage{thmtools}
\usepackage{thm-restate}

\usepackage[colorlinks=true,
            linkcolor=black!70!blue,
            citecolor=black!70!green,
            urlcolor=black!70!blue,
            pdftitle={Universal Properties of Petri Net Unfoldings},
            pdfauthor={Serge Lechenne, Hugo Paquet}]{hyperref}
\usepackage[capitalise,noabbrev]{cleveref}

\theoremstyle{plain}
\newtheorem{theorem}{Theorem}[section]
\newtheorem{lemma}[theorem]{Lemma}
\newtheorem{proposition}[theorem]{Proposition}
\newtheorem{corollary}[theorem]{Corollary}

\theoremstyle{definition}
\newtheorem{definition}[theorem]{Definition}
\newtheorem{example}[theorem]{Example}

\theoremstyle{remark}
\newtheorem{remark}[theorem]{Remark}

\usepackage[affil-it]{authblk}

\title{\bfseries Universal Properties of Petri Net Unfoldings}

\author[ ]{Serge Lechenne\thanks{\texttt{serge.lechenne@inria.fr}, \url{https://orcid.org/0009-0001-9626-4742}}}
\author[ ]{Hugo Paquet\thanks{\texttt{hugo.paquet@inria.fr}, \url{https://orcid.org/0000-0002-8192-0321}}}
\affil[ ]{Inria, \'Ecole Normale Sup\'erieure -- PSL University, CNRS, France}

\date{}

\tikzset{
  proplace/.style = {place, fill=blue!20, minimum size=5mm, inner sep=0pt},
  trans/.style    = {transition, scale=0.6, fill=black!20},
  open/.style     = {coordinate},
  ptoken/.style   = {token, scale=1, text=white, font=\tiny,
                     inner sep=0pt, minimum size=2.6mm},
  procbox/.style  = {rounded corners=2.5pt, fill=black!7, draw=none, inner sep=0pt},
  mor/.style      = {-{Stealth[length=2mm,width=1.7mm]}, thick, red!80!black,
                     shorten >=1.5pt, shorten <=1.5pt},
  catlab/.style   = {font=\small},
}

\tikzset{
  conflict/.style={decorate,
                   decoration={snake, amplitude=0.2mm, segment length=1mm},
                   color=blue!70!black, thick},
  causality/.style={->, thick, -{Triangle[open]}},
  every token/.style={scale=0.8},
}

\newlength{\hlength}
\newlength{\vlength}
\newcommand{\calN}{\mathcal{N}}

\newcommand{\calE}{\mathcal{E}}

\newcommand{\calJ}{J}

\newcommand{\calU}{\mathcal{U}}
\newcommand{\calM}{\mathcal{M}}

\newcommand{\Occ}{\mathbf{Occ}}
\newcommand{\Petri}{\mathbf{Petri}}

\newcommand{\WGPetri}{\mathbf{Petri}^{\mathsf{WG}}}

\newcommand{\WGOcc}{\mathbf{Occ}^{\mathsf{WG}}}

\newcommand{\ES}{\mathbf{Ev}}
\newcommand{\SES}{\mathbf{EvSym}}

\newcommand{\U}{\mathsf{U}}
\newcommand{\Q}{\mathsf{Q}}
\newcommand{\Pre}{\mathrm{Pre}}
\newcommand{\Post}{\mathrm{Post}}
\newcommand{\Set}{\mathbf{Set}}

\newcommand{\Con}{\mathrm{Con}}
\newcommand{\conf}{\mathscr{C}}
\newcommand{\lt}{\ensuremath{<}}
\newcommand{\id}{\mathrm{id}}
\newcommand{\profto}{\mathrel{\ooalign{$\longrightarrow$\cr\hfil\!\raisebox{-1.2pt}{\rotatebox{90}{$-$}}\hfil\cr}}}

\newcommand{\Colim}{\colim}
\DeclareMathOperator*{\colim}{colim}
\newcommand{\inb}{\mathrm{in}}
\newcommand{\outb}{\mathrm{out}}
\newcommand{\Path}{\mathbf{Path}}

\newcommand{\Cat}{\mathbf{Cat}}

\newcommand{\semWG}{\llbracket - \rrbracket^{\mathsf{WG}}}
\newcommand{\semWGA}[1]{\llbracket #1 \rrbracket^{\mathsf{WG}}}
\newcommand{\semS}{\llbracket - \rrbracket^{\mathrm{Sym}}}
\newcommand{\semSA}[1]{\llbracket #1 \rrbracket^{\mathrm{Sym}}}

\allowdisplaybreaks

\begin{document}
\maketitle

\begin{abstract}
\noindent
It is an established idea in concurrency theory that every Petri net admits an
unfolding semantics. This is a denotational object that represents its domain of
possible executions. Unfoldings play an important role in practical analysis and
verification.

This paper is concerned with the following well-known problem: while the
unfolding resembles a universal construction in the category of Petri nets, it
generally fails to satisfy the expected universal property. This is because the
unfolding construction overlooks the net's internal symmetries.

There are two solutions: make these symmetries explicit to obtain a weak
universal property (one that holds only ``up to symmetry''); or break the
symmetries by assigning individual identities to components of the net. We
review these two solutions and establish, in each case, a universal unfolding of
Petri nets to event structures.

This paper demonstrates a 2-categorical approach to Petri net unfoldings. We
show that each unfolding semantics determines a 2-categorical relative
adjunction involving Petri nets and event structures. Viewed in this way, the
above two constructions can be related formally via an appropriate morphism of
adjunctions. We exhibit a 2-density property of event structures which implies
that unfolding functors are essentially unique.
\end{abstract}

\medskip

\section{Introduction}
\label{sec:intro}

The theory of Petri nets has long helped to analyse the operational behaviour of concurrent systems in various application domains \cite{24143,winskel1987petri,PNR}. But this line of research has primarily targeted a restricted class of well-behaved Petri nets satisfying a ``safety'' condition. Safety is a helpful restriction in practice but it feels mathematically ad-hoc. This paper contributes to the categorical theory of unrestricted (or \emph{unsafe}) Petri nets, pioneered by \cite{meseguer1990petri,FMOP} and continuing to attract interest \cite{GENUN,baez2020open,WGPN,baez2021categories}. 

Specifically, we consider the unfolding problem for  Petri nets, which has attracted attention due to the symmetry issues that arise in the absence of safety (e.g. \cite{meseguer1997semantics}). We re-examine two approaches, respectively by Hayman--Winskel \cite{GENUN} and by Kock \cite{WGPN}, which exemplify the two main solutions to the unfolding problem: embrace the symmetries, or break them. 

We explore this from a 2-categorical perspective, and show that event structures \cite{ES1} continue to be an appropriate semantic domain for Petri nets, even in the unsafe setting, provided the symmetries are handled adequately.

\subsection{Petri nets unfoldings and event structures}

Informally, a Petri net is a graph with two types of nodes, places
and transitions, equipped with a collection of symbolic placeholders
(tokens) distributed across the places. For example, the following is a Petri net
where we have drawn places as circles, transitions as squares, and tokens as
bullets:
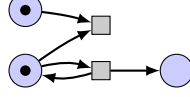
\[
\begin{tikzpicture}
  \node[place,tokens=1, scale = 0.5, fill=blue!20]   (p1) at (-1, -0.4) {};
  \node[place,tokens=1, scale = 0.5, fill=blue!20]   (p2) at (-1, 0.4) {};
  \node[place,tokens=0, scale = 0.5, fill=blue!20]   (p3) at (1, -0.4) {};
  \node[transition, scale=0.6, fill=black!20] (t1) at (0, 0.2) {} ;
  \node[transition, scale=0.6, fill=black!20] (t2) at (0, -0.4) {} ;
  \draw[-latex,thick, bend left=5] (p1) to (t1) {} ;
   \draw[-latex,thick, bend left=15] (p1) to (t2) {} ;
   \draw[-latex,thick, bend left=5] (p2) to (t1) {} ;
   \draw[-latex,thick] (t2) to (p3) {} ;
   \draw[-latex, thick, bend left=15] (t2) to (p1) {} ;
\end{tikzpicture}
\]
The assignment of tokens to the places of a Petri net is called a marking. This
is the only aspect of the net that changes during execution, via the
activation (or \emph{firing}) of transitions, consuming tokens and producing new
ones. For our example, here are some possible firings:
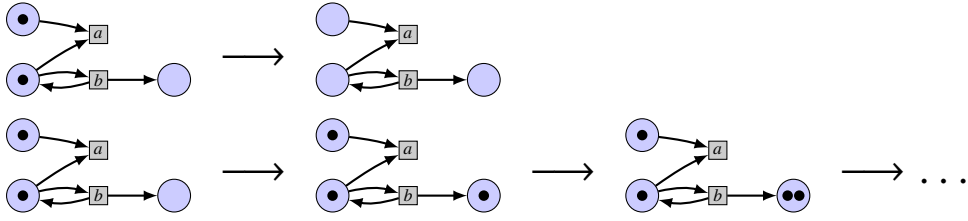
\begin{align*}
  &
\begin{tikzpicture}
  \node[place,tokens=1, scale = 0.5, fill=blue!20]   (p1) at (-1, -0.4) {};
  \node[place,tokens=1, scale = 0.5, fill=blue!20]   (p2) at (-1, 0.4) {};
  \node[place,tokens=0, scale = 0.5, fill=blue!20]   (p3) at (1, -0.4) {};
  \node[transition, scale=0.6, fill=black!20] (t1) at (0, 0.2) {$a$} ;
  \node[transition, scale=0.6, fill=black!20] (t2) at (0, -0.4) {$b$};
  \draw[-latex,thick, bend left=5] (p1) to (t1) {} ;
   \draw[-latex,thick, bend left=15] (p1) to (t2) {} ;
   \draw[-latex,thick, bend left=5] (p2) to (t1) {} ;
   \draw[-latex,thick] (t2) to (p3) {} ;
   \draw[-latex, thick, bend left=15] (t2) to (p1) {} ;
    \end{tikzpicture}
 \quad \raisebox{1em}{\scalebox{1.5}{$\longrightarrow$}} \quad
 \begin{tikzpicture}
  \node[place,tokens=0, scale = 0.5, fill=blue!20]   (p1) at (-1, -0.4) {};
  \node[place,tokens=0, scale = 0.5, fill=blue!20]   (p2) at (-1, 0.4) {};
  \node[place,tokens=0, scale = 0.5, fill=blue!20]   (p3) at (1, -0.4) {};
\node[transition, scale=0.6, fill=black!20] (t1) at (0, 0.2) {$a$} ;
  \node[transition, scale=0.6, fill=black!20] (t2) at (0, -0.4) {$b$};
  \draw[-latex,thick, bend left=5] (p1) to (t1) {} ;
   \draw[-latex,thick, bend left=15] (p1) to (t2) {} ;
   \draw[-latex,thick, bend left=5] (p2) to (t1) {} ;
   \draw[-latex,thick] (t2) to (p3) {} ;
   \draw[-latex, thick, bend left=15] (t2) to (p1) {} ;
\end{tikzpicture}
  \\
  &
\begin{tikzpicture}
  \node[place,tokens=1, scale = 0.5, fill=blue!20]   (p1) at (-1, -0.4) {};
  \node[place,tokens=1, scale = 0.5, fill=blue!20]   (p2) at (-1, 0.4) {};
  \node[place,tokens=0, scale = 0.5, fill=blue!20]   (p3) at (1, -0.4) {};
\node[transition, scale=0.6, fill=black!20] (t1) at (0, 0.2) {$a$} ;
  \node[transition, scale=0.6, fill=black!20] (t2) at (0, -0.4) {$b$};
  \draw[-latex,thick, bend left=5] (p1) to (t1) {} ;
   \draw[-latex,thick, bend left=15] (p1) to (t2) {} ;
   \draw[-latex,thick, bend left=5] (p2) to (t1) {} ;
   \draw[-latex,thick] (t2) to (p3) {} ;
   \draw[-latex, thick, bend left=15] (t2) to (p1) {} ;
    \end{tikzpicture}
 \quad \raisebox{1em}{\scalebox{1.5}{$\longrightarrow$}} \quad
 \begin{tikzpicture}
  \node[place,tokens=1, scale = 0.5, fill=blue!20]   (p1) at (-1, -0.4) {};
  \node[place,tokens=1, scale = 0.5, fill=blue!20]   (p2) at (-1, 0.4) {};
  \node[place,tokens=1, scale = 0.5, fill=blue!20]   (p3) at (1, -0.4) {};
  \node[transition, scale=0.6, fill=black!20] (t1) at (0, 0.2) {$a$} ;
  \node[transition, scale=0.6, fill=black!20] (t2) at (0, -0.4) {$b$};
  \draw[-latex,thick, bend left=5] (p1) to (t1) {} ;
   \draw[-latex,thick, bend left=15] (p1) to (t2) {} ;
   \draw[-latex,thick, bend left=5] (p2) to (t1) {} ;
   \draw[-latex,thick] (t2) to (p3) {} ;
   \draw[-latex, thick, bend left=15] (t2) to (p1) {} ;
 \end{tikzpicture}
  \quad \raisebox{1em}{\scalebox{1.5}{$\longrightarrow$}} \quad
 \begin{tikzpicture}
  \node[place,tokens=1, scale = 0.5, fill=blue!20]   (p1) at (-1, -0.4) {};
  \node[place,tokens=1, scale = 0.5, fill=blue!20]   (p2) at (-1, 0.4) {};
  \node[place, scale = 0.5, fill=blue!20]   (p3) at (1, -0.4) {}
   [children are tokens, token distance=0.4em]
  child {node [token] {}}
  child {node [token] {}};
  \node[transition, scale=0.6, fill=black!20] (t1) at (0, 0.2) {$a$} ;
  \node[transition, scale=0.6, fill=black!20] (t2) at (0, -0.4) {$b$};
  \draw[-latex,thick, bend left=5] (p1) to (t1) {} ;
   \draw[-latex,thick, bend left=15] (p1) to (t2) {} ;
   \draw[-latex,thick, bend left=5] (p2) to (t1) {} ;
   \draw[-latex,thick] (t2) to (p3) {} ;
   \draw[-latex, thick, bend left=15] (t2) to (p1) {} ;
\end{tikzpicture}  
\quad \raisebox{1em}{\scalebox{1.5}{$\longrightarrow \dots$}} \quad 
\end{align*}

The dynamic behaviour of a Petri net can be intricate. For instance, the firing
of a transition may produce enough tokens to enable other firings; this defines
a causality relation between firings. Meanwhile several transitions can compete
for the same tokens, which creates conflict and nondeterminism.

To reason about this complex behaviour, a denotational semantics is useful, and this is what event
structures (introduced by Plotkin, Nielsen and Winskel in a seminal 1981 paper
about Petri net unfoldings \cite{ES1}) aspire to provide. The rough idea is to
represent all possible sequences of transitions in a single structure that
accounts for causal dependency, conflict, and concurrency. The semantics of our
example net would be the infinite event structure below, in which each event
represents the firing of a single transition (indicated here by the event label);
an arrow represents a causal dependency; and a wavy blue line indicates a
conflict.
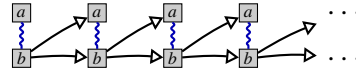
\[
\begin{tikzpicture}
\node[transition, scale=0.6, fill=black!20] (a1) at (0, 0.2) {$a$} ;
\node[transition, scale=0.6, fill=black!20] (b1) at (0, -0.4) {$b$};
\node[transition, scale=0.6, fill=black!20] (a2) at (1, 0.2) {$a$} ;
  \node[transition, scale=0.6, fill=black!20] (b2) at (1, -0.4) {$b$};
\node[transition, scale=0.6, fill=black!20] (a3) at (2, 0.2) {$a$} ;
  \node[transition, scale=0.6, fill=black!20] (b3) at (2, -0.4) {$b$};
\node[transition, scale=0.6, fill=black!20] (a4) at (3, 0.2) {$a$} ;
\node[transition, scale=0.6, fill=black!20] (b4) at (3, -0.4) {$b$};
\node[] (a5) at (4.3, 0.2) {$\dots$} ;
  \node[] (b5) at (4.3, -0.4) {$\dots$};
  \draw[conflict] (a1) to (b1) {} ;
  \draw[conflict] (a2) to (b2) {} ;
  \draw[conflict] (a3) to (b3) {} ;
  \draw[conflict] (a4) to (b4) {} ;
\draw[causality, bend left=5] (b1) to (b2) {} ;
\draw[causality, bend left=5] (b1) to (a2) {} ;
\draw[causality, bend left=5] (b2) to (b3) {} ;
\draw[causality, bend left=5] (b2) to (a3) {} ;
\draw[causality, bend left=5] (b3) to (b4) {} ;
\draw[causality, bend left=5] (b3) to (a4) {} ;
\draw[causality, bend left=5] (b4) to (b5) {} ;
\draw[causality, bend left=5] (b4) to (a5) {} ;
\end{tikzpicture}
\]
Note that conflict is hereditary, e.g. the first event labelled $a$ is
implicitly in conflict with all events labelled $b$, since no execution can see both events. We will define event
structures formally in \S\ref{sec:event-structures}. Informally, $b$ can fire repeatedly and forever but
as soon as $a$ fires the execution must stop.

This kind of semantics is known as an \emph{unfolding} (the cycle in the original Petri
net is `unfolded' to an infinite chain in the event structure). Event structures
are closely connected to a special class of acyclic Petri nets known as
occurrence nets (\S\ref{subsec:occ}), thus the terminology `unfolding' also refers to the
process of turning an arbitrary Petri net into an occurrence net.

\subsection{Universality issues in the unfolding of Petri nets}

Unfolding semantics can sometimes be ambiguous or non-canonical. The simplest problematic situation is when, at some stage of execution, a transition can choose between several tokens.
\subparagraph{A problematic example.}
\label{example-intro} Consider the net
\[
  \begin{tikzpicture}
    \node[place, scale = 0.5, fill=blue!20]   (p1) at (-1, 0) {}
     [children are tokens, token distance=0.4em]
     child {node [token] {}}
     child {node [token] {}};
\node[transition, scale=0.6, fill=black!20] (a1) at (0, 0) {$a$} ;
\draw[-latex,thick] (p1) -- (a1)  ;
\end{tikzpicture}
\]
whose marking allows for two successive firings. A natural candidate for its unfolding is the event structure
\[
  \begin{tikzpicture}
\node[transition, scale=0.6, fill=black!20] (a1) at (0, 0) {$a$} ;
\node[transition, scale=0.6, fill=black!20] (a1) at (0.8, 0) {$a$} ;
\end{tikzpicture}
\]
comprising two events without any conflict or dependency. But there is a mismatch: in the event structure, we can distinguish between two single-firing processes, whereas the Petri net has just one possible single-step execution:
\[
  \begin{tikzpicture}
    \node[place, scale = 0.5, fill=blue!20]   (p1) at (-1, 0) {}
     [children are tokens, token distance=0.4em]
     child {node [token] {}}
     child {node [token] {}};
\node[transition, scale=0.6, fill=black!20] (a1) at (0, 0) {$a$} ;
\draw[-latex,thick] (p1) -- (a1)  ;
\end{tikzpicture}
\qquad \raisebox{0em}{\scalebox{1.5}{$\longrightarrow$}} \qquad
  \begin{tikzpicture}
    \node[place, scale = 0.5, fill=blue!20]   (p1) at (-1, 0) {}
     [children are tokens, token distance=0.4em]
      child {node [token] {}};
\node[transition, scale=0.6, fill=black!20] (a1) at (0, 0) {$a$} ;
\draw[-latex,thick] (p1) -- (a1)  ;
\end{tikzpicture}
\]
Thus the event structure above fails to satisfy the universal property that one
might naturally expect of an unfolding: that the executions of a Petri net
correspond bijectively with the processes of its unfolding. (We will explain in
\S\ref{subsec:ex-paths} why this is a universal property in the categorical sense. See also \cite{GENUN}.)

This failure of universality is well-known, and it is common in Petri net theory
to impose a further restriction: that every reachable marking has at most one
token per place. Petri nets satisfying this condition are traditionally called
\emph{safe} nets. Safe nets \emph{do} admit a universal unfolding \cite{ES1}.

In this paper we do not make any safety assumptions and instead consider
 existing proposals for recovering a universal unfolding in the general case.

 \subparagraph{Unfolding general Petri nets to event structures.} To resolve this issue,
 two methods have emerged. The first method consists in enriching the unfolding with
 additional information expressing that certain processes are `the same'. This
 idea culminates in the work of Hayman and Winskel, based on a theory of
 explicit symmetries \cite{GENUN, SPN}.

 The second method is to modify the very definition of Petri net so that
 elements (tokens and edges) carry an individual identity. Viewed in this style, the
 problematic net above admits two processes
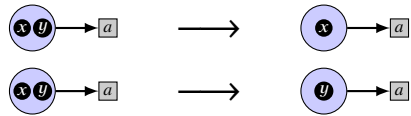
\begin{align*}
&  \begin{tikzpicture}
    \node[place, scale = 0.7, fill=blue!20]   (p1) at (-1, 0) {}
     [children are tokens, token distance=0.7em]
     child {node [token,scale = 1.3] {$x$}}
     child {node [token,scale = 1.3] {$y$}};
\node[transition, scale=0.6, fill=black!20] (a1) at (0, 0) {$a$} ;
\draw[-latex,thick] (p1) -- (a1)  ;
\end{tikzpicture}
\qquad \raisebox{0.4em}{\scalebox{1.5}{$\longrightarrow$}} \qquad
  \begin{tikzpicture}
    \node[place, scale = 0.7, fill=blue!20]   (p1) at (-1, 0) {}
     [children are tokens, token distance=0.5em]
      child {node [token, scale = 1.3] {$x$}};
\node[transition, scale=0.6, fill=black!20] (a1) at (0, 0) {$a$} ;
\draw[-latex,thick] (p1) -- (a1) ;
\end{tikzpicture} \\
&    \begin{tikzpicture}
    \node[place, scale = 0.7, fill=blue!20]   (p1) at (-1, 0) {}
     [children are tokens, token distance=0.7em]
     child {node [token,scale = 1.3] {$x$}}
     child {node [token,scale = 1.3] {$y$}};
\node[transition, scale=0.6, fill=black!20] (a1) at (0, 0) {$a$} ;
\draw[-latex,thick] (p1) -- (a1)  ;
\end{tikzpicture}
\qquad \raisebox{0.4em}{\scalebox{1.5}{$\longrightarrow$}} \qquad
  \begin{tikzpicture}
    \node[place, scale = 0.7, fill=blue!20]   (p1) at (-1, 0) {}
     [children are tokens, token distance=0.5em]
      child {node [token, scale = 1.3] {$y$}};
\node[transition, scale=0.6, fill=black!20] (a1) at (0, 0) {$a$} ;
\draw[-latex,thick] (p1) -- (a1) ;
\end{tikzpicture}
\end{align*}
distinguished by the specific token that is consumed. This is consistent with the
event structure above. This method is formalized using the theory of
whole-grain Petri nets developed by Kock \cite{WGPN}.

\subsection{Objectives and contributions}

This paper has two objectives. First, we revisit each method separately and
establish in both cases a universal unfolding of Petri nets to event structures.
The constructions given by Hayman--Winskel \cite{GENUN} and Kock \cite{WGPN} are
currently limited to occurrence nets, which are not adequate as a semantic
domain (roughly speaking, because places can be redundant). We will see however that connecting to event structures requires care. The second objective is to connect the two methods, following the intuitive idea that
forgetting names induces new symmetries.

 Making all of this precise requires some technical constructions. Universal
 properties are stated in the language of category theory, and in this paper we
 additionally need some basic 2-category theory. This is essential because
 the unfolding of Petri nets is universal only in a 2-categorical sense. (We
 note that 2-categorical methods turn up already in \cite{GENUN,WGPN,ESS}. They are a natural tool for dealing with named elements and symmetries.)

 We briefly outline our results and the organization of this paper. In
 \S\ref{sec:categ-petri-nets}, we recall that Petri nets form a category
 $\Petri$ and introduce a 2-category $\WGPetri$ of whole-grain Petri nets
(following Kock \cite{WGPN}).  We relate them using a functor
 $\Q : \WGPetri \to \Petri$ which forgets the names of tokens (and edges) to recover mere multiplicities. 
 
 In \S\ref{subsec:occ}, we introduce the
 sub-category $\Occ$ of occurrence nets, which can be defined equivalently either
 in whole-grain style or in ordinary style. We recall the definition of the category $\ES$ of event structures and show that the unfolding of
         whole-grain Petri nets induces a 2-adjunction between $\WGPetri$ and $\ES$. This is closely related to the unfolding of whole-grain
         Petri nets to occurrence nets found in  \cite{WGPN} but the link to event
         structures requires the introduction of new, more general class of morphisms and therefore a new unfolding functor.

In \S\ref{sec:unf-sem-sym}, we revisit the unfolding of ordinary Petri nets. Our main contribution is to establish a \emph{relative adjunction} \cite{RM} involving the
         categories $\Petri$ and $\ES$, and a 2-category $\SES$ of \emph{event
         structures with symmetry} \cite{ESS}. Furthermore, we show that
         the embedding $\ES \to \SES$ is a 2-dense 2-functor, which characterizes the unfolding uniquely up to symmetry.

Finally, in \S\ref{sec:adjunction-morphism}, we formally  connect the two universal constructions (ordinary and whole-grain) using a morphism of relative pseudo-adjunctions.

\subparagraph*{Note on terminology and related work.} Petri net theory is a vast area and the word `unfolding' is sometimes used for other purposes. Here we are concerned only with the process of unfolding a Petri net to an occurrence net or an event structure. We also note that the symmetry issues we discuss are similar to issues in the representation of unmarked nets as certain monoidal categories. We refer the reader to \cite{WGPN} for a thorough comparison.

\subsection{Elements and multiplicities}
\label{subsec:elem-mult}

An important point for this paper is the distinction between multisets of
elements of a set $X$, and sets of named occurrences of elements of $X$. This
section is just to fix notation for this basic principle.

\subparagraph{Multisets.} If $X$ is any set, let $\calM(X)$ denote the set of multisets of elements of $X$. Formally these are defined as functions $X \to \mathbb N$ with finite support. 

Now, for a set $S$ equipped with a function $m : S \to X$, we call
$\Q(m)$ the function $X \to \mathbb{N} \cup \{\infty \}$ giving for each $x \in X$ the cardinality of $m^{-1}\{ x\}$. Under simple conditions on the fibres of $m$, we have that $\Q(m) \in \calM(X)$. We regard $\Q$ as a forgetful operation: while
$S$ may contain several named occurrences of an element $x \in X$, in $\Q(m)$
their identity is erased and the copies are indistinguishable.

\subparagraph{Multirelations.}

A span of sets $X \xleftarrow f S \xrightarrow g Y$
can equivalently be presented by a function $\langle f, g \rangle : S \to X \times Y$. Under similar cardinality conditions on the fibres, via the multiplicity counting-operation $\Q$ this induces a \emph{multirelation} $\Q(\langle f, g \rangle )$, seen as element of $\calM(X \times Y)$.

Spans can be composed (via pullback) and form a bicategory (e.g. \cite{BSR}). When infinite coefficients are allowed multirelations are also closed under an infinite form of matrix multiplication, and the multiplicity-count operation is functorial.  (The multirelations in this paper are all finitary, and indeed satisfy further conditions to ensure that finiteness is preserved by composition.)

We will overload the variable $\Q$ and use it to denote the functor mapping
whole-grain Petri nets to ordinary Petri nets. This is appropriate because
the functor consists in applying $\Q$ to every component. 

\section{Categories of Petri nets, and multiplicity count}
\label{sec:categ-petri-nets}

We start by presenting two categories of Petri nets. First we look at the
traditional kind of Petri net with the standard notion of morphisms based on
multirelations \cite{GENUN}. Then we consider the recent theory of whole-grain
Petri nets \cite{WGPN},
based on spans instead of multirelations. We will see that the whole-grain framework is
inherently 2-categorical.

The section is organized as follows. In \S\ref{subsec:nets} we define each kind
of net and we see that a whole-grain Petri net induces an ordinary net if
individual identities are forgotten. Then in \S\ref{subsec:morphisms} we
consider morphisms between nets. This gives a category of Petri nets and a
2-category of whole-grain Petri nets, respectively $\Petri$ and $\WGPetri$, and
a $2$-functor $\Q: \WGPetri \rightarrow \Petri$.
\subsection{Ordinary Petri nets and whole-grain Petri nets}
\label{subsec:nets}

We recall the classical definition of Petri nets, based on multisets
and multirelations.

\begin{definition}[Petri net]
  A Petri net $P$ consists of:
  \begin{itemize}
  \item a set $S$ of \emph{places} and a set $T$ of \emph{transitions};
  \item multirelations $\Pre : S \profto T$ and $\Post : S \profto T$ such that for every $t\in T$ there are finitely many $s \in S$ such that $\Pre(s, t) > 0$ and finitely many such that $\Post(s, t) > 0$;
  \item a multiset $\mu \in \calM(S)$ of places, the \emph{marking};
  \end{itemize}
satisfying the following two conditions:
\begin{itemize}
    \item The net is \emph{grounded}: for every $t \in T$, there is $s \in S$ such that $\Pre(s,
      t) > 0$.
    \item The net has no \emph{isolated places}: for every $s \in S$, there exists $t \in T$ such that $\Pre(s,t) > 0$ or $\Post(s, t) > 0$.
\end{itemize}
\end{definition}
The marking $\mu$ specifies the number of tokens in each place at
initialisation. The fact that $\Pre$ and $\Post$ are multirelations
indicates that several tokens may pass through a single edge
during one step of execution. This multiplicity-based formalism follows the
\emph{collective-token philosophy} \cite{van2009configuration}: the tokens in each place are
indistinguishable from each other, have no individual identity, and
cannot be individually tracked as they move through the net.

It is common to impose an additional safety condition that forbids having
multiple tokens in one place \cite{ES2} for any reachable marking. We do not need to define this, precisely because this paper is about the universality issues with general ``unsafe'' Petri nets.

We now look at whole-grain Petri nets \cite{WGPN}.
\begin{definition}[Whole-grain Petri net]
  \label{def:whole-grain}
  A whole-grain Petri net $P$ consists of:
  \begin{itemize}
  \item a set $S$ of places and a set $T$ of transitions;
  \item two sets $I$ and $O$ together with spans $S \leftarrow I \rightarrow T$
  and $S \leftarrow O \rightarrow T$;
    \item a finite set $M$ together with a map $M \to S$, the \emph{marking};
  \end{itemize}
such that the functions $I \to T$ and $O \to T$ have finite fibres. Additionally, we require that the maps $I \rightarrow T$ and $I + O \rightarrow S$ are
surjective.
 
\end{definition}
The surjectivity conditions correspond to the two conditions imposed
on ordinary Petri nets \cite{ES2}, as is also the finite fibres hypothesis (that states each place/transition has a finite number of transitions/places connecting to it via $I$ and $O$). This formalism allows for a much more precise description of
the token game where tokens are individually tracked as they are consumed and produced by transitions: this is the \emph{individual-token philosophy}.
A basic observation is that we can always forget (or quotient out) the
token and edge identities, to recover an ordinary Petri net.

\begin{proposition}
  \label{prop:Qonobjects}
  For a whole-grain Petri net $P$, with components labelled as in
Definition~\ref{def:whole-grain}, there is an ordinary Petri net $\Q(P)$ with transition set $T$, place set $S$, marking $\mu = \Q(M)$, and multirelations
$\Pre = \Q(S \leftarrow I \rightarrow T)$ and $\Post = \Q(S \leftarrow O \rightarrow T)$. This Petri net is denoted $\Q(P).$
\end{proposition}

In other words, $\Q(P)$ is a version of $P$ in which parallel edges are combined
into a single edge with multiplicity and the initial marking is reduced to a
multiplicity count for each place. The operational behaviour of the nets $P$ and $Q(P)$ correspond; e.g., a transition $e$ is enabled in $P$ if and only if it is enabled as a transition of $Q(P)$. 

\begin{remark}
Both kinds of Petri nets in this paper are assumed to be grounded and to have no isolated places. These two conditions are common but not systematically imposed, so this deserves a brief comment. In the context of net unfoldings, the grounded-ness condition seems essential, to avoid tokens being uncontrollably produced. However, one could likely do away with the condition on isolated places, at the cost of a slightly more complex 2-category of whole-grain Petri nets (because Theorem~\ref{thm:discreteness} below would not hold). 
\end{remark}

\subsection{Morphisms of Petri nets}
\label{subsec:morphisms}

Maps of Petri nets have a well-established theory; they have a canonical justification as relations preserving the token game \cite{winskel1984new}. 
\begin{definition}
  \label{def:morphisms-petri}
  For $P = (S, T, \Pre, \Post, \mu)$ and
  $P' = (S', T', \Pre', \Post', \mu')$, a \emph{map of Petri nets}
  $P \to P'$ consists of a function $\eta : T \rightarrow T'$  and a multirelation $\beta : S \profto S'$ such
  that the diagrams
\[\begin{tikzcd}[row sep=0.2em]
	& S \\
	1 \\
	& {S'}
	\arrow["\beta"{inner sep=.8ex}, "\shortmid"{marking}, from=1-2, to=3-2]
	\arrow["\mu"{inner sep=.8ex}, "\shortmid"{marking}, from=2-1, to=1-2]
	\arrow["{\mu'}"'{inner sep=.8ex}, "\shortmid"{marking}, from=2-1, to=3-2]
\end{tikzcd}
\quad
\begin{tikzcd}
	S & T \\
	{S'} & {T'}
	\arrow["\Pre"{inner sep=.8ex}, "\shortmid"{marking}, from=1-1, to=1-2]
	\arrow["\beta"'{inner sep=.8ex}, "\shortmid"{marking}, from=1-1, to=2-1]
	\arrow["\eta", from=1-2, to=2-2]
	\arrow["{\Pre'}"'{inner sep=.8ex}, "\shortmid"{marking}, from=2-1, to=2-2]
\end{tikzcd}\quad
\begin{tikzcd}
	S & T \\
	{S'} & {T'}
	\arrow["\Post"{inner sep=.8ex}, "\shortmid"{marking}, from=1-1, to=1-2]
	\arrow["\beta"'{inner sep=.8ex}, "\shortmid"{marking}, from=1-1, to=2-1]
	\arrow["\eta", from=1-2, to=2-2]
	\arrow["{\Post'}"'{inner sep=.8ex}, "\shortmid"{marking}, from=2-1, to=2-2]
\end{tikzcd}
\]
commute. (The multisets $\mu$ and $\mu'$ are seen as multirelations from a singleton set, and the function $\eta$ is seen as a multirelation with
$\eta(t, t') = 1$ if $\eta(t) = t'$ and $0$ otherwise.)
\end{definition}
\begin{remark}
Equivalent to the above diagrammatic conditions are the following equations, for every $p' \in P'$ and $t \in T$:
\begin{itemize}
    \item $\sum_{p \in P} \beta(p, p')\mu(p) = \mu'(p')$;
    \item $\sum_{p \in P}\beta(p, p') \Pre(p, t) = \Pre'(p',\eta(t))$; and 
    \item $\sum_{p \in P}\beta(p, p')\Post(p, t) = \Post'(p',\eta(t))$.
\end{itemize}
\end{remark}

In a map of Petri nets, the multirelation $\beta$ can be read as a reverse action on places: every
place $s' \in S'$ that admits an edge to a transition $\eta(t)$ must be
assigned a multiset of places incoming for $t$. When $\beta$ arises from a total function $S' \to S$, the map $(\eta, \beta)$ is called a \emph{folding map}. Morphisms of Petri nets compose (as multirelations) to form a category that we denote $\Petri$.

We now turn to the morphisms of whole-grain Petri nets, closely following
Kock \cite{WGPN}. This involves a bit more data (in particular,
because in the whole-grain setting, the multirelation $\beta$ must be
replaced by another span) but axioms are stated as simple pullback
conditions.

\begin{definition}
  \label{def:map-whole-grain}
  Suppose that $P = (S,T,I,O, M)$ and $P' = (S', T', I', O', M')$ are
  whole-grain Petri nets. A \emph{map of (whole-grain) Petri nets}
  $\varphi : P \rightarrow P'$ is a family of five functions
  relating $P$ and $P'$ componentwise such that
\[\begin{tikzcd}
	S & I & T & O & S & M & S \\
	{S'} & {I'} & {T'} & {O'} & {S'} & {M'} & {S'}
	\arrow[from=1-1, to=2-1]
	\arrow["{(a)}"{description}, draw=none, from=1-1, to=2-2]
	\arrow[from=1-2, to=1-1]
	\arrow[from=1-2, to=1-3]
	\arrow[from=1-2, to=2-2]
	\arrow["{(b)}"{description}, draw=none, from=1-2, to=2-3]
	\arrow[from=1-3, to=2-3]
	\arrow["{(c)}"{description}, draw=none, from=1-3, to=2-4]
	\arrow[from=1-4, to=1-3]
	\arrow[from=1-4, to=1-5]
	\arrow[from=1-4, to=2-4]
	\arrow["{(d)}"{description}, draw=none, from=1-4, to=2-5]
	\arrow[from=1-5, to=2-5]
	\arrow[from=1-6, to=1-7]
	\arrow[from=1-6, to=2-6]
	\arrow["{(e)}"{description}, draw=none, from=1-6, to=2-7]
	\arrow[from=1-7, to=2-7]
	\arrow[from=2-2, to=2-1]
	\arrow[from=2-2, to=2-3]
	\arrow[from=2-4, to=2-3]
	\arrow[from=2-4, to=2-5]
	\arrow[from=2-6, to=2-7]
\end{tikzcd}\] commutes. (We keep the individual functions anonymous for
clarity; the functions $S \to S'$ are all the same in the diagram.) 

The only relevant maps for
this paper are those satisfying further conditions, as follows.
\begin{itemize}
\item A map $\varphi$ is an \emph{\'etale map} if $(b)$ and $(c)$ are
  pullback squares and the map $M \to M'$ is an identity function (in
  particular $M = M'$).
  \item A map $\varphi$ is a \emph{cabling map} if $(a)$, $(d)$, and $(e)$ are
        pullback squares and the map $T \to T'$ is an identity function (in
        particular $T' = T$).
\end{itemize}
\end{definition}

More informally, \'etale maps are those preserving the input and output arities of transitions up to names: the pullback condition enforces this via an appropriate isomorphism of fibres
A cabling is in some sense \emph{place-\'etale}, as it preserves arities of places. (One could relax the identity map axiom to an invertibility condition, but the definition above makes the overall formalism simpler \cite{WGPN}.)

\begin{remark}
  Our cabling maps generalise those of Kock \cite{WGPN}, who only considered
  maps between Petri nets having the same initial marking. It is important for
  our purposes to allow for varying markings (subject to the pullback condition
  $(e)$---compare with the first condition in
  Definition~\ref{def:map-whole-grain})  for a proper
  connection with event structures.
\end{remark}

More concretely, the pullback conditions for an \'etale map assert that a
transition and its image must have the same number of incoming and outgoing
edges, with the map providing a specific bijection. A cabling map has the
analogous property for places, and additionally the two Petri nets must have the
same transitions. By combining \'etale maps and cabling maps we recover a
whole-grain version of the maps of ordinary Petri nets in Definition~\ref{def:morphisms-petri}.

\begin{definition}
  For whole-grain Petri nets $P$ and $P'$, a \emph{rational map}
  $P \to P'$ is a span
  \[
     P \xleftarrow{\varphi} R \xrightarrow{\psi} P'
  \]
  where $R$ is another whole-grain net, $\varphi$ is a cabling map and $\psi$ is
  an \'etale map.
\end{definition}

We recover \'etale maps as the class of rational maps for which $\varphi = \id$.
The components of a rational map can be explicitly laid out as
\[\begin{tikzcd}
	{S_P} & {I_P} & {T_P} & {O_P} & {S_P} & M \\
	{S_R} & {I_R} & {T_R} & {O_R} & {S_R} & {M'} \\
	{S_{P'}} & {I_{P'}} & {T_{P'}} & {O_{P'}} & {S_{P'}} & {M'}
	\arrow[from=1-2, to=1-1]
	\arrow[from=1-2, to=1-3]
	\arrow[from=1-4, to=1-3]
	\arrow[from=1-4, to=1-5]
	\arrow[from=1-6, to=1-5]
	\arrow[from=2-1, to=1-1]
	\arrow[from=2-1, to=3-1]
	\arrow["\lrcorner"{anchor=center, pos=0.125, rotate=180}, draw=none, from=2-2, to=1-1]
	\arrow[from=2-2, to=1-2]
	\arrow[from=2-2, to=2-1]
	\arrow[from=2-2, to=2-3]
	\arrow[from=2-2, to=3-2]
	\arrow["\lrcorner"{anchor=center, pos=0.125}, draw=none, from=2-2, to=3-3]
	\arrow[equals, from=2-3, to=1-3]
	\arrow[from=2-3, to=3-3]
	\arrow[from=2-4, to=1-4]
	\arrow["\lrcorner"{anchor=center, pos=0.125, rotate=90}, draw=none, from=2-4, to=1-5]
	\arrow[from=2-4, to=2-3]
	\arrow[from=2-4, to=2-5]
	\arrow["\lrcorner"{anchor=center, pos=0.125, rotate=-90}, draw=none, from=2-4, to=3-3]
	\arrow[from=2-4, to=3-4]
	\arrow[from=2-5, to=1-5]
	\arrow[from=2-5, to=3-5]
	\arrow["\lrcorner"{anchor=center, pos=0.125, rotate=180}, draw=none, from=2-6, to=1-5]
	\arrow[from=2-6, to=1-6]
	\arrow[from=2-6, to=2-5]
	\arrow[from=3-2, to=3-1]
	\arrow[from=3-2, to=3-3]
	\arrow[from=3-4, to=3-3]
	\arrow[from=3-4, to=3-5]
	\arrow[equals, from=3-6, to=2-6]
	\arrow[from=3-6, to=3-5]
\end{tikzcd}\]
\begin{proposition}
 \label{prop:Qonmor}
  For $P, P' \in \WGPetri$, every rational map $P \xleftarrow{\varphi} R \xrightarrow{\psi} P'$
  induces a morphism of ordinary Petri nets $(\eta, \beta): \Q(P) \to \Q(P')$ where $\eta$ is the
  function $T_R \to T_{P'}$ and $\beta$ is the finitary multirelation
  $\Q(S_P \leftarrow S_R \to S_{P'})$.
\end{proposition}
\begin{proof}
We verify the first condition in  Definition~\ref{def:morphisms-petri} and omit the other two, which use similar arguments. Explicitly, we must show that for every $p' \in S_{P'}$, $\sum_{p \in S_P} \beta(p, p')\mu(p) = \mu'(p')$, where $\mu$ and $\mu'$ are the respective markings of $\Q(P)$ and $\Q(P')$. 

Recall that, by definition of $\Q$, for $p \in S_P$ and $p' \in S_{P'}$, $\beta(p, p')$ is the cardinality of the set $\varphi_S^{-1}\{p\} \cap \psi_S^{-1}\{p'\}$; $\mu(p)$ is the cardinality of the fibre $m_P^{-1}\{ p\}$ for the marking function $m_P : M_P \to S_P$; and similarly $\mu'(p')$ is the cardinality of $m_{P'}^{-1}\{p'\}$. 

By the axioms of rational maps we have the following situation:
\[
\begin{tikzcd}
	S_P & S_R & S_P' \\
	M_P & {M_R} & {M_{P'}} 
	\arrow[from=1-2, to=1-1, "\varphi_S"']
	\arrow[from=1-2, to=1-3, "\psi_S"]
	\arrow[from=2-1, to=1-1, "m_P"]
	\arrow["\lrcorner"{anchor=center, pos=0.125, rotate=180}, draw=none, from=2-2, to=1-1]
	\arrow[from=2-2, to=1-2, "m_R"]
	\arrow[from=2-2, to=2-1]
	\arrow[equals, from=2-2, to=2-3]
	\arrow[from=2-3, to=1-3, "m_{P'}"]
\end{tikzcd}
\]
Thus, by the characterization of pullbacks in $\mathbf{Set}$ we may identify $M_R$ with the coproduct $\sum_{p \in S_P} \varphi_S^{-1}\{p\} \times m_P^{-1}\{p \}$, and under this identification the map $m_R : M_R \to S_R$ is given by $(p, r, j) \mapsto r$. For $p' \in S_{P'}$ we obtain the desired equation by considering the bijection 
\begin{align*}
m_{P'}^{-1}\{p'\} &\cong m_R^{-1} \left( \psi_S^{-1} \{ p' \} \right) \\
&\cong \{ (p, r, j) \mid  p \in S_P, r \in \varphi_S^{-1}\{p \}, j \in m_P^{-1}\{p\} \text{ and } r \in \psi_S^{-1} \{p'\} \} \\ 
&= \{ (p, r, j) \mid  p \in S_P, r \in \varphi_S^{-1}\{p \} \cap \psi_S^{-1}\{p'\}, j \in m_P^{-1}\{p\} \} \\ 
&\cong \coprod_{p \in P} \left(\varphi_S^{-1}\{p \} \cap \psi_S^{-1}\{p'\}\right) \times m_P^{-1}\{p\}
\end{align*}
whose domain has cardinality $\mu(p')$ and whose codomain has cardinality $\sum_{p \in S_P} \beta(p, p') \times \mu(p)$.

\end{proof}

\'Etale maps correspond to folding maps:
\begin{lemma}
\label{lemma:folding}
    For $P, P' \in \WGPetri$, if $\psi : P \rightarrow P'$ is an \'etale map (seen as a degenerate rational map), then $\Q(\psi)$ is a folding map. 
\end{lemma}
\begin{proof}
For a span of sets of the form $S_P =\!\!=\!\!= S_P \to S_{P'}$ the multirelation $S_P \profto S_{P'}$ given by $\Q$ is a total function.
\end{proof}

Rational maps $(\varphi, \psi) : P \to P'$ and $(\varphi', \psi') : P' \to P''$ compose, following the composition of spans as pullbacks:
\[\begin{tikzcd}[row sep=1em, column sep=1em]
	&& {R''} && \\
	& R && {R'} \\
	P && {P'} && {P''}
	\arrow[dashed, from=1-3, to=2-2]
	\arrow[dashed, from=1-3, to=2-4]
	\arrow["\lrcorner"{anchor=center, pos=0.125, rotate=-45}, draw=none, from=1-3, to=3-3]
	\arrow["\varphi"', from=2-2, to=3-1]
	\arrow["\psi", from=2-2, to=3-3]
	\arrow["{\varphi'}"', from=2-4, to=3-3]
	\arrow["{\psi'}", from=2-4, to=3-5]
\end{tikzcd}\]
Here we must take care when discussing pullbacks since the two legs of a rational map belong to different classes of maps. So, formally, the pullbacks are taken in the category $\WGPetri_{\mathrm{gen}}$ of whole-grain Petri nets and arbitrary maps (that is, the maps of Definition~\ref{def:map-whole-grain}, without any pullback conditions), which corresponds to taking pullbacks componentwise in $\Set$. 
That the composite span gives a rational map is ensured by the following lemma:
\begin{lemma}
	\label{lem:pullback-cabling-etale}
 For whole-grain Petri nets $P', R, R'$, if $\psi : R \rightarrow P'$ is a cabling map, and $\varphi : R' \rightarrow P'	$ is an \'etale map, then there is a $\WGPetri_{\mathrm{gen}}$-pullback square as below for which $\psi''$ is a cabling map and $\varphi''$ is an \'etale map.
 \[\begin{tikzcd}[row sep=1em, column sep=1em]
	&& {R''} && \\
	& R && {R'} \\
	 && {P'} &&
	\arrow["{\varphi''}"', from=1-3, to=2-2]
	\arrow["{\psi''}", from=1-3, to=2-4]
	\arrow["\lrcorner"{anchor=center, pos=0.125, rotate=-45}, draw=none, from=1-3, to=3-3]
	\arrow["\psi"', from=2-2, to=3-3]
	\arrow["{\varphi'}", from=2-4, to=3-3]
\end{tikzcd}\]
\end{lemma}
\begin{proof}
 The diagram
\[\begin{tikzcd}[row sep=0.8em, column sep=0.8em]
	& {I_R} && {I_{P'}} \\
	{S_{R}} && {S_{P'}} \\
	& {I_{R''}} && {I_{R'}} \\
	{S_{R''}} && {S_{R'}}
	\arrow[from=1-2, to=1-4]
	\arrow[from=2-1, to=1-2]
	\arrow[from=2-1, to=2-3]
	\arrow[from=2-3, to=1-4]
	\arrow[from=3-2, to=1-2]
	\arrow[from=3-2, to=3-4]
	\arrow[from=3-4, to=1-4]
	\arrow[from=4-1, to=2-1]
	\arrow[from=4-1, to=3-2]
	\arrow[from=4-1, to=4-3]
	\arrow[from=4-3, to=2-3]
	\arrow[from=4-3, to=3-4]
\end{tikzcd}\]
commutes in $\Set$, by definition of maps in $\WGPetri_{\mathrm{gen}}$, with the front and back squares pullbacks. The right-hand square is a pullback because $\varphi'$ is an \'etale map, and therefore by the pasting law for pullbacks the left-hand square is also a pullback. This is one of the conditions for $\varphi''$ to be \'etale; other conditions are derived in the same way. 

The proof that $\psi''$ is a cabling map is also analogous, but for a minor technical point: to preserve the cabling property we use that pullbacks of identities can be taken to be identities. (Indeed one can make a global choice of pullbacks in $\Set$ that satisfies this.) We omit the proof. 
\end{proof}

This composition operation for rational maps gives rise to a bicategory. There is an identity rational map for every Petri net $P$ given by the identity span $P =\!= P =\!= P$. Pullbacks are only defined up to isomorphism and so the composition of spans is
only associative and unital up to coherent invertible 2-cells. The 2-cells are
defined as standard morphisms of spans.

\begin{definition}
  For whole-grain Petri nets $P$ and $P'$ and rational maps as on the left below,
\[\begin{tikzcd}[row sep=0.3em]
	P & R & {P'} \\
	P & R & {P'}
	\arrow["\varphi"', from=1-2, to=1-1]
	\arrow["\psi", from=1-2, to=1-3]
	\arrow["{\varphi'}"', from=2-2, to=2-1]
	\arrow["{\psi'}", from=2-2, to=2-3]
\end{tikzcd}\qquad\qquad \qquad 
    \begin{tikzcd}[row sep=0.3em]
	& {R} \\
	P && {P'} \\
	& R'
	\arrow["\varphi"', from=1-2, to=2-1]
	\arrow["\psi", from=1-2, to=2-3]
	\arrow["\alpha", from=1-2, to=3-2]
	\arrow["{\varphi'}", from=3-2, to=2-1]
	\arrow["{\psi'}"', from=3-2, to=2-3]
\end{tikzcd}\]
  a \emph{$2$-cell} $(\varphi, \psi) \to (\varphi', \psi')$ is a morphism $\alpha : R \rightarrow R'$ that makes the above triangles commute.
  \end{definition}
  Whole-grain Petri nets, rational maps and $2$-cells assemble into a
  bicategory that we denote $\WGPetri$. That we have a bicategory and not a
  category is typical for span-like morphisms and this is closely connected to
  the fact that the whole-grain setting tracks occurrences of elements with
  individual names, not just multiplicities. Any renaming of elements will give
  a 2-cell, and in fact \emph{every} 2-cell in this bicategory is a renaming:

\begin{restatable}[Discreteness]{theorem}{discreteness}
\label{thm:discreteness}
  For whole-grain Petri nets $P$ and $P'$ and rational maps
  $(\varphi, \psi), (\varphi', \psi') : P \rightarrow P'$, there is at most one
  $2$-cell $(\varphi, \psi) \rightarrow (\varphi', \psi')$ and when it exists it is invertible.
\end{restatable}
\begin{proof}
Appendix~\ref{sec:discreteness-proof}.
\end{proof}

In other words, the bicategory $\WGPetri$ is locally essentially discrete. This property
prevents any coherence issues and greatly eases the proofs of several theorems
below.

We have described two categorical models for Petri nets, a bicategory $\WGPetri$ and a category $\Petri$. We now connect them using the multiplicity-counting operation $\Q$ (of \S\ref{subsec:elem-mult}), whose action on whole-grain Petri nets extends to a 2-functor $\WGPetri \to \Petri$ (this is a straightforward verification of the axioms). For this statement to make sense, $\Petri$ is understood as a locally discrete $2$-category.
\begin{theorem}
There is a 2-functor $\Q : \WGPetri \to \Petri$ extending the action on objects and morphisms described in Proposition~\ref{prop:Qonobjects} and Proposition~\ref{prop:Qonmor}.
\end{theorem}
\begin{proof}
  It remains only to deal with the $2$-cells. A $2$-cell $\alpha : R \to R'$ between rational maps $P \leftarrow R \rightarrow P'$ and $P \leftarrow R' \rightarrow P'$  implies that $T_R = T_{R'}$ and the functions $T_P =\!= T_R \to T_P'$ and $T_P =\!= T_R \to T_{P'}$ are the same. We have seen that $\alpha$ has a component $S_R \to S_{R'}$ which is an isomorphism of spans between $S_P \leftarrow S_R \rightarrow S_{P'}$ and $S_P \leftarrow S_{R'} \rightarrow S_{P'}$. Necessarily these spans also have the same image under $\Q$. The axioms for a 2-functor follow immediately.
\end{proof}

We note that this functor is surjective on objects and morphisms.  
\begin{proposition} 
	\label{prop:Q-surjective}
	For every $P \in \Petri$ there exists $P' \in \WGPetri$ with $\Q(P') = P$, and for every morphism $(\eta, \beta) : P_1 \to P_2$ of ordinary Petri nets there is a rational map $(\psi, \varphi) : P_1' \to P_2'$ with $\Q(\psi, \varphi) = (\eta, \beta)$.
\end{proposition}
\begin{proof}[Proof sketch]
    For the places and transitions of $P'$, take the corresponding sets in $P$. Then let $I := \left\{(p, t, i) \mid 1 \leq i \leq \Pre(p, t)\right\}$ and $O := \left\{(p, t, i) \mid 1 \leq i \leq \Post(p, t)\right\}$, with span legs given by suitable projections out of these sets. Finally let $M:= \{(p, i) \mid 1 \leq i \leq \mu(p) \}$ with the obvious projection to places. We have defined $P'$ with $\Q(P') = P$.

	Now let $P_1, P_2 \in \Petri$ and define whole-grain nets $P_1', P_2'$ with $\Q(P_i') = P_i$ as in the previous paragraph. Given a morphism $(\eta, \beta): P_1 \rightarrow P_2$,
	 build a whole-grain Petri net $R$ with the components $T_R := T_{P_1}$, $S_R :=  \left\{(p, p', k) \mid p \in S_{P_1}, p' \in S_{P_2}, 1 \leq k \leq \beta(p, p')\right\}$, $I_R := \left\{(p, p', k, t, i) \mid (p, p', k) \in S_R, t \in T_R, 1 \leq i \leq \Pre_{P_1}(p, t) \right\}$, 
 $O_R := \{(p, p', k, t, i) \mid (p, p', k) \in S_R, t \in T_R, 1 \leq i \leq \Post_{P_1}(p, t) \}$, and $M_R = M_{P_2}$, equipped with appropriate projection maps. Then build a span $P_1' \xleftarrow{\psi} R \xrightarrow{\varphi} P_2'$ in terms of projections, chosen bijections and the function $\eta$. We omit the rest of the proof including the verification that $\Q(\psi, \varphi) = (\eta, \beta)$.
\end{proof}    

We emphasize that $\Q : \WGPetri \to \Petri$ is not an equivalence: the whole-grain Petri net
  \[
  \begin{tikzpicture}
    \node[place, scale = 0.5, fill=blue!20]   (p1) at (-1, 0) {}
     [children are tokens, token distance=0.5em]
     child {node [token] {$x$}}
     child {node [token] {$y$}};;
\node[transition, scale=0.6, fill=black!20] (a1) at (0, 0) {$a$} ;
\draw[-latex,thick,bend left=15] (p1) to (a1)  ;
\draw[-latex,thick,bend right=15] (p1) to (a1)  ;
\end{tikzpicture}
\]
  has four distinct (rational) automorphisms, since the two tokens and the two edges can be permuted, however its image under $\Q$ only has an identity automorphism.
 
\section{The token game: paths, occurrence nets, and event structures}
\label{subsec:occ}

The ``token game'' specifies the operational behaviour of a Petri net. There is
only one rule: when a transition is fired, the marking is modified according to
the incoming and outgoing edges for that transition. Slightly more formally, in an ordinary Petri net, a transition $t$ is enabled if, for every place $p$, $\mu(p) \geq \Pre(p, t)$, and the firing of an enabled transition has the action of updating $\mu$ to a new marking $\mu'$ given by $\mu'(p) = \mu(p) - \Pre(p, t) + \Post(p, t)$ for every $p$.

But this simple rule creates complex dynamics: several transitions may be enabled at a given point and firing one might disable the other or enable new transitions. So a Petri net generally admits many possible execution behaviours (\emph{cf.} the example in \S\ref{example-intro}), including some in which firings occur in parallel.

For whole-grain Petri nets, the complexities are the same, and in addition one must track token and edge identities. 

In this section we recall the formal notion of execution path (or \emph{process}) for a Petri net, and we then define occurrence nets which, as we recall, represent colimits of paths. We compare the whole-grain view and the traditional view on paths: the latter is based on causal nets \cite{ES1} and the former uses directed acyclic graphs
\cite{WGPN}.

\subsection{The execution paths of a Petri net}
\label{subsec:ex-paths}

One key idea of \cite{WGPN} is to organize the transitions of a given execution path
into a directed acyclic graph, where a node represents a transition firing and an edge
between two firings indicates that a token is produced by one and consumed by
the other (via a place). Technically this gives an ``open-ended'' graph with dangling edges, because the tokens in the initial marking are not produced by any firings, and the tokens in the final marking are not consumed by any firings. 

Recall that a directed graph can be defined as a pair of sets $N, A$ (nodes and arcs) together with source and target functions $s, t: A \to N$. Graphs in this paper are defined in this style but they are open-ended, grounded, and acyclic: 
\begin{definition}
  An \emph{open-ended graph} $G$ consists of a pair of finite sets $(N, A)$
  of nodes and arcs, together with partial functions $s : A \rightharpoonup N$ and $t : A \rightharpoonup N$,
  specifying the source and target nodes of each arc where they are defined. Say $G$ is \emph{grounded} if $t$ is surjective, and \emph{acyclic} if the only directed path from a node to itself has length 0. 

  The \emph{in-boundary} of $G$ is the subset of $A$ where $s$ is undefined, and
  the \emph{out-boundary} is the subset of $A$ where $t$ is undefined. These
  subsets are denoted $\inb(G)$ and $\outb(G)$, respectively.
\end{definition}

We observe, following \cite{WGPN}, that graphs can be seen as special whole-grain Petri nets, if one temporarily drops the assumption that Petri nets should have no isolated places\footnote{Kock \cite{WGPN} does not impose this condition on whole-grain nets. This axiom is not required for the basic theory, but saves us a lot of coherence trouble via Theorem~\ref{thm:discreteness}. Temporarily dropping it when considering graphs causes no issues.}.
Indeed the data of an (open-ended, acyclic, grounded) graph $G = (N, A, s, t)$ can be written out as $A \hookleftarrow I \ \xrightarrow{t} N \xleftarrow{s}\ O \hookrightarrow A$ where $I$ is the complement of $\inb(G)$ in $A$ and $O$ is the complement of $\outb(G)$ in $A$. The initial marking is given by the inclusion map $\inb(G) \hookrightarrow A$.

A graph $G$ is a whole-grain Petri net and therefore also gives rise to an ordinary Petri net $\Q(G)$. In this instance the operation $\Q$ does not actually discard any information because a graph has at most a single token per place and no parallel edges. A Petri net of the form $\Q(G)$ is called a \emph{causal net} (introduced using different language in \cite{ES1}).

\begin{definition}
  For $P \in \WGPetri$, a \emph{path} of $P$ is a graph $G$  together with an \'etale map $p : G \to P$. A morphism of paths $(G, p) \to (G', p')$
  is defined as an \'etale map $\psi : G \to G'$ such that $p = p' \circ \psi$. We
  write $\Path(P)$ for the \emph{category of paths of $P$}.
\end{definition} 

\begin{definition}
  For $P \in \Petri$, a \emph{path} of $P$ is a graph $G$   together with a folding map $p : \Q(G) \to P$. A morphism of paths $(G, p) \to (G', p')$
  is defined as an \'etale map $f : G \to G'$ such that $p = p' \circ \Q(f)$. We
  write $\Path(P)$ for the \emph{category of paths of $P$}.
\end{definition}

Maps of Petri nets induce functors between categories of paths:
\begin{lemma}[\cite{WGPN}]
	\label{lem:functoriality}
	Let $P$ and $P'$ be whole-grain Petri nets.
	\begin{enumerate}
\item An \'etale map $\varphi : P \to P'$ induces a functor $\Path(P) \to \Path(P')$ given by post-composition: $(G, p : G \to P)$ is mapped to $(G, \varphi \circ p$).
\item A cabling map $\psi : P' \to P$ induces a functor $\Path(P) \to \Path(P')$ defined by pullback along $\psi$. A path $(G, p : G \to P)$ is mapped to $(\psi^*G, \psi^*G \to P')$ given by 
\[\begin{tikzcd}[row sep=1.4em, column sep=1em]
	P & {P'} \\
	G & {\psi^*G}.
	\arrow["\psi"',from=1-2, to=1-1]
	\arrow["p",from=2-1, to=1-1]
	\arrow["\lrcorner"{anchor=center, pos=0.125, rotate=180}, draw=none, from=2-2, to=1-1]
	\arrow[from=2-2, to=1-2]
	\arrow[from=2-2, to=2-1]
\end{tikzcd}\]
	\end{enumerate}
Combining the two, a rational map $P \leftarrow R \rightarrow P'$ induces a functor $\Path(P) \to \Path(P')$.
\end{lemma}

Note that Point (2) in Lemma~\ref{lem:functoriality} uses a special case of Lemma~\ref{lem:pullback-cabling-etale}.

\subsection{Occurrence nets and unfoldings}
\label{subsec:occ-unfold}

\subsubsection{Occurrence nets}

Occurrence nets are a special class of Petri nets. The motivation is that an occurrence net should represent, as a single `unfolded' net, the full domain of paths of a Petri net. But for ordinary nets it is difficult to state this formally. We will first define occurrence nets directly, and later discuss how they represent domains of paths. 

There are various equivalent definitions of occurrence nets, but in this paper it makes sense to use one based on causality and conflict relations, which makes plain the connection to event structures.
\begin{definition}[\cite{ES1}]
\label{def:occurrence-net}
  Let $P = (S, T, \Pre, \Post, \mu)$ be a Petri net with `no parallel edges': the multirelations $\Pre$ and $\Post$ are ordinary relations. Say two
  transitions $t, t' \in T$ are in \emph{immediate conflict}, written
  $t \mathbin{\#_0} t'$, if there exists a place $s$ with $\Pre(s, t)$ and
  $\Pre(s, t')$. Define binary relations $<$ and $\#$ on the set $S \uplus T$ as
  follows:
  \begin{itemize}
    \item $<$ is the transitive closure of  $\Pre \cup \Post$, so $u < v$ iff there is a non-empty path from $u$ to $v$ in the graph underlying $P$;
    \item $\#$ is the \emph{hereditary closure} of $\#_0$ under ${<}$, that is,
          the smallest relation containing $\#_0$ and such that if $u \mathbin{\#} u'$ and
          $u < v$ then $v \mathbin{\#} u'$.
        \end{itemize}
        The Petri net $P$ is an \emph{occurrence net} if: 
		\begin{itemize}
			\item the relations $<$ and
        $\#$ are irreflexive; 
		\item each place $s$ has at most one $t$ with $\Post(s, t)$;
		\item the set $\{ u \mid u < v \}$ is finite for every
        $v \in S \uplus T$; and
		\item the multiset $\mu$ is a set (at most one token
        per place) containing precisely the $\leq$-minimal places: those with no incoming transitions. This set is also called the \emph{in-boundary} of $P$. 
		\end{itemize}
  \end{definition}
Let $\Occ$ be the full subcategory of $\Petri$ spanned by occurrence nets. 

A \emph{whole-grain occurrence net} is a whole-grain Petri net $O$ such that $\Q(O)$ is an occurrence net.\footnote{Kock calls this an \emph{occurrence hypergraph}; his
    definition is equivalent but emphasises other aspects \cite{WGPN}.} Occurrence nets and rational maps form a bicategory, called $\WGOcc$, which turns out to be equivalent to $\Occ$.
\begin{theorem}
  The restricted 2-functor $\Q : \WGOcc \rightarrow \Occ$ forms an equivalence. 
\end{theorem}
\begin{proof}[Proof sketch]
  It follows easily from Proposition~\ref{prop:Q-surjective} that is is surjective on objects and morphisms. This implies in particular that for every $O, O' \in \WGOcc$, the functor (of essentially discrete categories) $\Q : \WGOcc(O, O') \to \Occ(\Q(O), \Q(O'))$ is surjective on objects. We prove it is an equivalence by showing it is full (it is automatically faithful). 
 When $O \leftarrow R \rightarrow O'$ and  $O \leftarrow R' \rightarrow O'$ are rational maps with the same image $(\eta, \beta)$ under $\Q$, by definition of cabling maps $R$ and $R'$ have the same transitions, and multiplicity conditions on occurrence nets ensure that the multirelation $\beta$ is  actually an ordinary relation, moreover satisfying $R \cong R' \cong \{ (p, p') \in S_O \times S_{O'} \mid \beta(p, p') \}$. From this we construct an appropriate 2-cell.
 \end{proof}

One key property is that an occurrence net $O$ (whole-grain or ordinary) is a canonical colimit of its paths: the map $\colim_{(G, p) \in \Path(O)} G \to O$ is an isomorphism of Petri nets, see \cite[Proposition 8.15]{WGPN}. 

We note also that occurrence nets are `closed under cabling' in the following sense:
\begin{lemma}
\label{lem:cabling-to-occ}
   Let $\psi : P \rightarrow P'$ be a cabling map. If $P'$ is an occurrence net, then so is $P$.
\end{lemma}
\begin{proof}[Proof note]
See \cite{WGPN} for a similar proof in case $P$ is a graph. The general proof is an easy argument by contradiction.
\end{proof}

\color{black}
\subsubsection{Unfolding of whole-grain Petri nets and \'etale maps}

One benefit of moving to the whole-grain approach (the key insight of \cite{WGPN}) is that it becomes easy to construct an unfolding with the expected universal property:
\begin{theorem}[\cite{WGPN}]
\label{ref:thmunfolding-etale}
  Let $\WGPetri_\textrm{\'et}$ and $\WGOcc_{\textrm{\'et}}$ denote the
  respective wide subcategories on \'etale morphisms. The inclusion functor $\WGOcc_{\textrm{\'et}} \hookrightarrow \WGPetri_{\textrm{\'et}}$ has a right adjoint $\U$. 
\end{theorem}

\begin{remark} 
	The unfolding $\U P$ of a whole-grain net $P$ is an occurrence net whose in-boundary corresponds to the token set of $P$. To see why this is the case, observe that the counit of this adjunction gives an \'etale map $\varepsilon_P : \U P \to P$, so in particular $\U P$ and $P$ have the same marking set $M_P$ (although they have different place sets). Since $\U P$ is an occurrence net, its marking and in-boundary coincide.
\end{remark}

The theorem states the following universal property: for every whole-grain Petri net $P$ and whole-grain occurence net $O$ with an \'etale map $\varphi : O \to P$, there is a unique map $O \to \U P$ that makes the diagram 
\[\begin{tikzcd}
	{\U P} & O \\
	& P
	\arrow["\varepsilon_P"',from=1-1, to=2-2]
	\arrow[dashed,from=1-2, to=1-1]
	\arrow["\varphi",from=1-2, to=2-2]
\end{tikzcd}\]
commute. In particular, from an \'etale map $P' \to P$ one obtains an \'etale map $\U P' \to U P$.

The generalization of this universal property to rational maps, i.e. establishing an adjunction between $\WGPetri$ and $\Occ$, is not immediate. This is what we do in the next section, leveraging Theorem~\ref{ref:thmunfolding-etale}.

\begin{corollary}
	For a whole-grain Petri net $P$, post-composition with $\varepsilon_P$ determines a functor $\Path(\U P) \to \Path(P)$ which is an isomorphism of categories. 
\end{corollary}
\begin{proof}
The inverse is induced by the universal property: any path $G \to P$ must factor through $\varepsilon_P$ via a unique path $G \to \U P$. Functoriality is a straightforward consequence of uniqueness. 
\end{proof}

\begin{remark}
Since occurrence nets are colimits of their paths, it also follows that there is a canonical isomorphism $\U P \cong \colim_{(G, p) \in \Path(P)} G$. This fits the initial motivation for occurrence nets and unfoldings. 
\end{remark}

\subsubsection{Unfolding of whole-grain Petri nets and rational maps}

First we state the main theorem of this section. (Note that we have overloaded the notation $\U$ to refer to the unfolding, but both functors $\U$ agree on their common domain $\WGPetri_{\textit{\'et}}$.)
\begin{theorem}
\label{thm:whole-grain-unfolding}
The embedding functor $\WGOcc \hookrightarrow \WGPetri$ admits a right adjoint, denoted $\U :\WGPetri \to \WGOcc$.
\end{theorem}

\begin{proof}
We will show that the counit of the restricted adjunction of Theorem~\ref{ref:thmunfolding-etale}, a family of \'etale maps $\varepsilon_P : \U P \to P$, extends to the rational setting. 

First recall that we can regard $\varepsilon_P$ as a degenerate rational map
 $\U P \xleftarrow{\id} \U P \xrightarrow{\varepsilon_P} P$, and it follows easily from properties of pullbacks that pre-composing this with another rational map $P_1 \xleftarrow{\psi} P_2 \xrightarrow{\varphi} \U P$ gives the span $P_1 \xleftarrow{\psi} P_2 \xrightarrow{\varepsilon_P \circ \varphi} P$. (This is a well-defined rational map: $\psi$ is a cabling map and both $\varphi$ and $\varepsilon_P$ are \'etale.)

We show that $(\id, \varepsilon_P)$ is a universal arrow in the following sense: for every $O \in \WGOcc$ and $P \in \WGPetri$, the post-composition functor 
\[
(\id, \varepsilon_P) \circ - : \WGOcc(O, \U P) \longrightarrow \WGPetri(O, P)
\]
is an equivalence of (essentially discrete) categories. It suffices to show it is essentially surjective and full, because faithfulness is immediate for functors between essentially discrete categories.

\emph{Essentially surjective.} Consider a rational map $O \xleftarrow{\psi} R \xrightarrow{\varphi} P$ (with $\varphi$ \'etale, $\psi$ cabling). By Lemma~\ref{lem:cabling-to-occ}, $R$ is an occurrence net. Thus we can instantiate Theorem~\ref{ref:thmunfolding-etale}: $\varphi$ must factor through $\varepsilon_P : \U P\to P$ via an \'etale map $\varphi' : R \to \U P$, so $\varphi = \varepsilon_P \circ \varphi'$. We have constructed a rational map $O \xleftarrow{\psi} R \xrightarrow{\varphi'} \U P$ in the pre-image of $(\varphi, \psi)$. 

\emph{Full.} Let $O \xleftarrow{\psi_1} R_1 \xrightarrow{\varphi_1} \U P$ and $O \xleftarrow{\psi_2} R_2 \xrightarrow{\varphi_2} \U P$ be rational maps such that $(\id, \varepsilon_P) \circ (\psi_1, \varphi_1)$ and $(\id, \varepsilon_P) \circ (\psi_2, \varphi_2)$ are related by a 2-cell $\alpha$ in $\WGPetri(O, P)$. We must show $(\psi_1, \varphi_1)$ and $(\psi_2, \varphi_2)$ are themselves related by a 2-cell. 

By Theorem~\ref{thm:discreteness}, the 2-cell $\alpha$ is an isomorphism $R_1 \to R_2$ such that
\[\begin{tikzcd}[row sep=tiny]
	& {R_1} & {\U P} & \\
	O &&& P \\
	& {R_2} & {\U P}
	\arrow["{\varphi_1}", from=1-2, to=1-3]
	\arrow["\alpha"{description}, from=1-2, to=3-2]
	\arrow["{\varepsilon_P}", from=1-3, to=2-4]
	\arrow["{\psi_1}", from=2-1, to=1-2]
	\arrow["{\psi_2}"', from=2-1, to=3-2]
	\arrow["{\varphi_2}"', from=3-2, to=3-3]
	\arrow["{\varepsilon_P}"', from=3-3, to=2-4]
\end{tikzcd}\]
commutes. But in particular $\alpha$ is \'etale, and so we must have $\varphi_1 = \varphi_2 \circ \alpha$ by the universal property of $\varepsilon_P$ (Theorem~\ref{ref:thmunfolding-etale}). Thus $\alpha$ also provides the required 2-cell $(\psi_1, \varphi_1) \to (\psi_2, \varphi_2)$.
\end{proof}

We have established the existence of the unfolding functor $\U$ by identifying the universal arrow $\varepsilon$. But it is interesting to look directly at the induced functorial action of $\U$ on rational maps. 
\begin{proposition}
Let $P \xleftarrow{\psi} R \xrightarrow{\varphi} P'$ be a rational map of whole-grain Petri nets. Up to isomorphism, the functor $\U$ maps $(\psi, \varphi)$ to a rational map \[ \U P \xleftarrow{\xi} R' \xrightarrow{\chi} \U P' \] 
consisting of the following components:
\begin{itemize}
	\item $R'$ is an occurrence net computed as \[ \Colim_{(G, p : G \to P) \in \Path(P)} \psi^* G \] or more explicitly as the colimit of the functor $\Path(P) \xrightarrow{\psi^*} \Path(R) \xrightarrow{(G, p) \mapsto G} \WGPetri_{\mathrm{gen}}$.
	\item $\xi$ is the canonical map out of a colimit \[ \Colim_{(G, p) \in \Path(P)} \psi^* G \quad \longrightarrow \quad \Colim_{(G, p) \in \Path(P)} G \] 
	induced by the family of pullback projections $\psi^* G \to G$.
	\item $\chi$ is the canonical map out of a colimit
\[
 \Colim_{(G, p) \in \Path(P)} \psi^* G \quad \longrightarrow \quad \Colim_{(H, p : H \to P') \in \Path(P')} H
 \] 
 induced by the family of colimit injections $\iota_{(\psi^* G, q)} : \psi^* G \to   \Colim_{(H, p) \in \Path(P')} H$ where $q$ consists of the (\'etale) pullback projection $\psi^* G \to R$ composed with $\varphi : R \to P'$. 
\end{itemize}
\end{proposition}
\begin{proof}
By construction of $\U$, the rational map $\U P \to \U P'$ induced by $(\psi, \varphi) : P \to P'$ is obtained by applying the universal property of $\varepsilon_{P'} : \U P' \to P'$ to the \'etale part of the rational map $(\psi, \varphi) \circ (\id, \varepsilon_P) :  \U P \to P'$. The latter is obtained by composition of spans using a pullback and therefore we are in the following situation:
\[\begin{tikzcd}[row sep=small, column sep=tiny]
	& {R'} &&& {\U P'} \\
	{\U P} && R && {P'} \\
	& P
	\arrow["\chi", dashed, from=1-2, to=1-5]
	\arrow["\xi"', from=1-2, to=2-1]
	\arrow["\pi", from=1-2, to=2-3]
	\arrow["\lrcorner"{anchor=center, pos=0.125, rotate=-45}, draw=none, from=1-2, to=3-2]
	\arrow["{\varepsilon_P}", from=1-5, to=2-5]
	\arrow["{\varepsilon_P}"', from=2-1, to=3-2]
	\arrow["\varphi"', from=2-3, to=2-5]
	\arrow["\psi", from=2-3, to=3-2]
\end{tikzcd}\]
In a presheaf category (such as $\WGPetri_\mathrm{gen}$) pullbacks commute with colimits, and $\U P$ is a colimit of the paths $G$ of $P$, so $R'$ is a colimit of the $\psi^* G$ and the characterization of $\xi$ follows easily. The characterization of $\chi$ is clear by the universal property of $\varepsilon_{P'}$. 
\end{proof}

Unfortunately, while things work smoothly in the whole-grain setting, unfoldings to occurrence nets are more difficult for ordinary Petri nets. The reason is that, for a net $P \in \Petri$, the projection functor $\Path(P) \to \Occ$ does not generally admit a colimit in $\Occ$.  (The colimit does exist if $P$ is a safe Petri net \cite{GENUN, ES2}.) 

\begin{example} Figure~\ref{fig:pathexample} contrasts the situation in $\WGPetri$ and $\Petri$ for the problematic example of \S\ref{example-intro}, here denoted $P$. (To read the figure, recall that graphs are special nets (\S\ref{subsec:ex-paths}), although we draw them with nodes and edges.) The symmetries in $\Path{(\Q P)}$ affect the universality of the potential unfolding of $\Q P$. The symmetries are not present at all in $\Path{(P)}$ because morphisms must respect token labels. 
\end{example}

\begin{figure}[t]
\centering
\begin{subfigure}[b]{0.60\textwidth}
\centering
\begin{tikzpicture}
  \proc{Aa}{(-0.58, 0.0)} {$x$}{$y$}{open} {open}
  \proc{Ab}{(-2.68, 2.0)} {$x$}{$y$}{trans}{open}
  \proc{Ac}{( 1.52, 2.0)} {$x$}{$y$}{open} {trans}
  \proc{Ad}{(-0.58, 4.0)} {$x$}{$y$}{trans}{trans}

  \draw[mor] (Aabox) -- (Abbox);
  \draw[mor] (Aabox) -- (Acbox);
  \draw[mor] (Abbox) -- (Adbox);
  \draw[mor] (Acbox) -- (Adbox);
\end{tikzpicture}
\caption{The category $\Path(P)$.}
\label{fig:pathexample-wg}
\end{subfigure}
\hfill
\begin{subfigure}[b]{0.36\textwidth}
\centering
\begin{tikzpicture}
  \proc{Ba}{(-0.58, 0.0)} {}{}{open} {open}
  \proc{Bb}{(-0.58, 2.0)} {}{}{trans}{open}
  \proc{Bd}{(-0.58, 4.0)} {}{}{trans}{trans}

  \draw[mor] (Babox) to[bend left=20]  (Bbbox);
  \draw[mor] (Babox) to[bend right=20] (Bbbox);
  \draw[mor] (Bbbox) to[bend left=20]  (Bdbox);
  \draw[mor] (Bbbox) to[bend right=20] (Bdbox);

  \draw[mor] (Bdbox.north east) to[out=52,in=128,looseness=1.3]
             (Bdbox.north west);
  \draw[mor] (Babox.south west) to[out=-128,in=-52,looseness=1.15]
             (Babox.south east);
\end{tikzpicture}
\caption{The category $\Path(\Q P)$.}
\label{fig:pathexample-ord}
\end{subfigure}
\caption{Skeletons of path categories for the example in \S\ref{example-intro}.}
\label{fig:pathexample}
\end{figure}
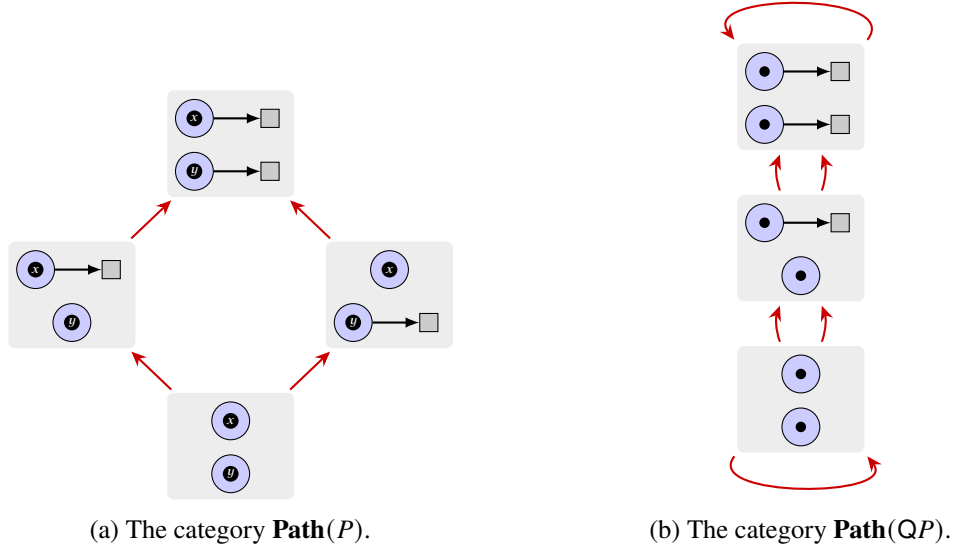

 One solution, found by Hayman and Winskel, is to enrich occurrence nets with additional structure encoding symmetry \cite{GENUN}, to render this problematic automorphism ``equivalent'' to the identity. With symmetry one can therefore state universal properties ``up to symmetry'', as will be developed in \S\ref{sec:unf-sem-sym}.

\subsection{Event structures}
\label{sec:event-structures}

Event structures \cite{ES1} are an abstraction of occurrence nets in which the places have been discarded: an event structure has a single set equipped with a partial order and conflict relation. These are designed to correspond to the relations on transitions in an occurrence net (\emph{cf.} Definition~\ref{def:occurrence-net}).

Our view is that, as a semantic domain for Petri nets, event structures are more appropriate than occurrence nets. They are easier to describe and reason about, and the only thing lost---the places---is arguably an irrelevant `implementation detail', since the execution of a Petri net is already fully described by the transitions firing.

\begin{definition}
An \emph{event structure} is a tuple $(E, {\leq}, \#)$ where $\leq$ is a partial order on $E$, and $\#$  is an irreflexive and symmetric binary relation on $E$, satisfying the axioms below:
\begin{itemize}
    \item For every $e \in E$  the set $\{e' \in E \mid e' \leq e \}$ is finite.
    \item For $e, e', e'' \in E$, if $e \leq_{E} e'$ and $e \mathrel{\#} e''$ then $e' \mathrel{\#} e''$.
\end{itemize}
\end{definition}
The execution paths of an event structure are called \emph{configurations}:
\begin{definition}
    A finite subset $X \subseteq E$ is called \emph{consistent} if no two events in $X$ are in conflict. It is called a \emph{configuration} if it is additionally down-closed. The set of consistent subsets of $E$ is denoted $\Con(E)$, and the set of configurations is called $\conf(E)$.
\end{definition}
Maps of event structures must faithfully preserve configurations:
\begin{definition}
    If $(E, \leq_E, \#_E)$ and $(E', \leq_{E'}, \#_{E'})$ are event structures, a function $f: E \rightarrow E'$ is a \emph{map of event structures} if for every $x \in \conf(E)$, the direct image $fx := \{f(e) \mid e \in x\}$ is in $\conf(E')$, and moreover the restriction of $f$ to $x$ is injective.
\end{definition}
Event structures and maps of event structures form a category $\ES$. Frequently, we will regard $\ES$ as a locally discrete $2$-category, just like we did for $\Petri$ and $\Occ$.

\newcommand{\E}[0]{\mathsf{E}}
We now recall (from \cite{ES1, ES2}) the relationship between occurrence nets and event structures. The set of transitions of an occurrence net, equipped with the relations $\leq$ and $\#$ of Definition~\ref{def:occurrence-net}, defines an event structure. This extends to a functor $\calE:\Occ \to \ES$. Conversely, for every event structure $E$ the pre-image $\calE^{-1}(E)$ is nonempty and indeed an object of $\calE^{-1}(E)$ can be constructed universally so as to give a functor $\calN: \ES\to \Occ$ and an adjunction
\begin{equation}
\label{eq:occ-es}
\begin{tikzcd}
	\Occ && \ES.
	\arrow[""{name=0, anchor=center, inner sep=0}, "\calE", shift left=2, from=1-1, to=1-3]
	\arrow[""{name=1, anchor=center, inner sep=0}, "\cal{N}", shift left=2, from=1-3, to=1-1]
	\arrow["\dashv"{anchor=center, rotate=90}, draw=none, from=1, to=0]
\end{tikzcd}
\end{equation}
\newcommand{\N}{\mathsf{N}}

By combining this adjunction with that of Theorem~\ref{thm:whole-grain-unfolding}, recalling the equivalence of $\WGOcc$ and $\Occ$, we obtain a universal unfolding of whole-grain Petri nets to event structures. 
 \begin{corollary}
\label{cor:whole-grain-unfolding}
There is a natural equivalence of setoids $\WGPetri(\N E, P) \simeq \ES(E, \semWGA{P})$ for every $P \in \WGPetri$ and $E \in \ES$; in other words, a pseudo-adjunction
\[ 
\begin{tikzcd}
	\WGPetri && \ES
	\arrow[""{name=0, anchor=center, inner sep=0}, "\semWG", shift left =2, from=1-1, to=1-3]
	\arrow[""{name=1, anchor=center, inner sep=0}, "\N", shift left=2,  from=1-3, to=1-1]
	\arrow["\dashv"{anchor=center, rotate=90}, draw=none, from=1, to=0]
\end{tikzcd}
\]
where 
$\N = \ES \xrightarrow{\calN} \Occ \simeq \WGOcc \hookrightarrow \WGPetri$, 
$\semWG = \WGPetri \xrightarrow{\mathsf{U}} \WGOcc \simeq \Occ \xrightarrow{\calE} \ES.$
\end{corollary}
(One could also define $\N$ explicitly, just by adapting the definition of $\calN$ to the whole-grain setting.)
In summary, the unfolding semantics of a Petri net is computed in two steps: first as an occurrence net, which combines all paths into a single (generally infinite) net, and then as an event structure. So far we have only done this for whole-grain Petri nets, giving the adjunction in Corollary~\ref{cor:whole-grain-unfolding}. We now turn to the more complex situation for ordinary nets. 
\section{Petri net unfoldings via explicit symmetries}
\label{sec:unf-sem-sym}

In this section we consider the unfolding of ordinary Petri nets. Our main contribution is a universal unfolding semantics in terms of event structures with symmetry. 

\subsection{Background: Petri net unfoldings to occurrence nets}
Unfoldings of Petri nets are often used in practice (e.g. \cite{UIP}) but it can be difficult to find an explicit construction. We now import a characterization from \cite{GENUN}. To give some intuition for the mutually recursive construction below, the occurrence net $\calU_P$ consists of the tokens of $P$ viewed as places, extended in a small-step fashion by considering reachable markings and transition occurrences. 

\begin{proposition}
\label{prop:winskel}
For a Petri net $P = (S_0, T_0, \Pre_0, \Post_0, \mu_0)$, there is a unique occurrence net $\calU_P = (S, T, \Pre, \Post)$ equipped with a folding map $\varepsilon_P = (\eta, \beta) : \calU_P \to P$, satisfying the equations below, in which `$\mathrm{co} \ A$' means that $A$ has no two places related by $<$ or $\#$,
\begin{align*}
    &S = \{(s, i) \mid s \in S_0, 0 \leq i \lt \mu_0(s)\} 
    \cup
    \{(t, s, i) \mid t \in T, s \in S_0, 0 \leq i < \Post_0(s, \eta(t)) \} \\
    &T = \{ (A, t) \mid A \subseteq S, t \in T_0, \mathrm{co}\ A, \text{and } | \{ s' \in A \mid \beta(s') =  s \}| = \Pre_0(s, t) \text{ for all } s \in S_0 \} 
\end{align*}
and with $\eta(A, t) = t$, $\beta(t, s, i) = \beta(s, i)= s$, $\Pre(s, (A, t))$ iff $s \in A$, and 
$\Post(s', (A, t))$ iff $s'$ is of the form $((A, t), s, i)$. (Recall that for an occurrence net the multirelations $\Pre$ and $\Post$ are relations, and the marking is uniquely determined.) 
\end{proposition}

This unfolding fails to satisfy the universal property required for an adjunction, essentially because in unsafe settings the small-step construction produces redundant data. For example, when the above construction is applied to the Petri net in \S\ref{example-intro}, two indistinguishable tokens become  two separate (distinguishable) places, leading to the symmetry issues already discussed there.

The unfolding does, in fact, satisfy the existence part of a universal property: 
\begin{lemma}[\cite{GENUN, ES2}]
\label{lemmawinskel}
Let $O \in \Occ$ and $P \in \Petri$. For every map  $(\eta, \beta): O  \rightarrow P$ of Petri nets, there exists a map $(\eta', \beta') : O \rightarrow \calU_P$ such that the diagram 
    \[\begin{tikzcd}
	O & P \\
	{\calU_P}
	\arrow["{(\eta,\beta)}", from=1-1, to=1-2]
	\arrow["{(\eta',\beta')}"', dashed, from=1-1, to=2-1]
	\arrow["{\varepsilon_P}"', from=2-1, to=1-2]
\end{tikzcd}\]
commutes. 
\end{lemma}
The insight of Hayman and Winskel \cite{GENUN} is that by tracking the internal symmetries of $\calU_P$ one can prove uniqueness up to a form of symmetry. This can be phrased as a universal property in a weak 2-categorical sense: a pseudo-adjunction. 

But adding symmetry to Petri nets, as in \cite{GENUN}, is a highly technical endeavour. There are several non-equivalent constructions, requiring nets with multiple markings, and the connection to event structures remains unclear because of a fundamental obstacle described in \cite[\S 7.1]{SPN}.

Our perspective is that these complications are unnecessary. We show how to bypass Petri nets with symmetries by directly unfolding to event structures with symmetry, which have a much simpler and established theory \cite{ESS}.  As we will see, this method relies on a 2-categorical relative adjunction and a 2-density property.

\subsection{The unfolding of a Petri net as an event structure with symmetry}

We first recall event structures with symmetry and the $2$-category $\SES$, originally introduced in \cite{ESS}.

\subsubsection{Event structures with symmetry}

\label{subsec:sem-def}
 Informally, symmetry on an event structure indicates which pairs of executions can be considered equivalent, via a bisimulation relation. 
\begin{definition}
    For an event structure $E$, an isomorphism family  is a set $\mathbb{S}$ of bijections $\theta:x \cong y$ between finite configurations of $x, y \in \conf(E)$, such that:
    \begin{itemize}
        \item $\mathbb{S}$ contains all identity bijections, and is stable under composition and inverse.  
        \item For $\theta : x \cong y \in \mathbb{S}$, if $x' \subseteq x$, then the restriction of $\theta$ to $x'$ is in $\mathbb{S}$.
        \item For $\theta : x \cong y \in \mathbb{S}$, if $x \subseteq x'$, then there exists an extension $y \subseteq y'$ and a bijection $\theta' : x' \cong y' \in \mathbb{S}$ such that $\theta'$ restricts to $\theta$.
    \end{itemize}
    The pair $(E, \mathbb{S})$ is called an \emph{event structure with symmetry}.
\end{definition}

We then consider symmetry-preserving maps:
\begin{definition}
    For event structures with symmetry $(E,\mathbb{S})$ and $(E', \mathbb{S}')$, a map $f : E \rightarrow E'$ \emph{preserves symmetry} if for all $\theta : x \cong y \in \mathbb{S}$, $f \theta : fx \cong fy \in \mathbb{S}'$, where $f\theta : f(e) \mapsto f(\theta(e))$.
\end{definition}
A novelty in the presence of symmetry is the following equivalence relation on maps:
\begin{definition}
    For event structures with symmetry $(E,\mathbb{S})$ and $(E', \mathbb{S}')$ and symmetry-preserving maps $f, g : E \rightarrow E'$, say that $f$ and $g$ are \emph{symmetric}, denoted $f \sim g$, if for every $x \in \conf(E)$ the bijection $fx \cong gx$ defined by $ f(e) \mapsto g(e)$ is in $\mathbb{S}'$.
\end{definition}
Together, event structures with symmetry, symmetry-preserving maps and symmetries of maps form a $2$-category $\SES$ \cite{ESS}. This is a degenerate 2-category in the sense that hom-categories are all \emph{setoids} (sets with an equivalence relation). There is an embedding $J : \ES \to \SES$. As we will see, the 2-categorical structure makes it possible to consider universal properties ``up to symmetry''.  

\subsubsection{The unfolding semantics with symmetry}

\newcommand{\semAlone}{\semSA{-}}
We proceed towards the construction of a 2-functor $\semAlone : \Petri \to \SES$. To build $\semSA{P}$ for a Petri net $P$, we equip the event structure $\calE(\calU_P)$ with an appropriate isomorphism family, which we define now.  The informal idea is that configurations corresponding to the same path in $P$ should be made symmetric. 

To make this precise we first need to formally construct the path corresponding to a configuration $x$ of $\calE(\calU_{P})$. Note that $x$ can itself be regarded as an event structure, inheriting the partial order from $\calE(\calU_{P})$, and with no conflict. There is an embedding $x \to \calE\calU_{P}$ which, because of the adjunction in \eqref{eq:occ-es}, corresponds to a morphism $\gamma_x : \calN x \rightarrow \calU_{P}$ (and the image $\calE(\gamma_x)$ is the inclusion map $x \to \calU_{P}$). Let $\alpha_x$ denote the composite map $\calN x \xrightarrow {\gamma_x} \calU_{P} \xrightarrow{\varepsilon_P} P$. The pair $(\calN x, \alpha_x)$ is a path  
of $P$, because the functor $\calN$ turns conflict-free event structures into graphs (see \cite{ES1}). An isomorphism of event structures $\theta : x \cong y$ between configurations of $\calE(\calU_P)$ induces a morphism of $\calU_P$-paths $\calN(\theta) : \calN x \cong \calN y$. The underlying bijection $\theta$ should be a symmetry if the paths $(\calN x, \alpha_x)$ and $(\calN y, \alpha_y)$ correspond to the same path in $P$. All of this is summarized in the commutative diagram 

\[\begin{tikzcd}
	{\calN x} && {\calN y} \\
	{\calU_P} && {\calU_P} \\
	& P
	\arrow["\cong"{description}, draw=none, from=1-1, to=1-3]
	\arrow["\theta"{description}, shift left=3, draw=none, from=1-1, to=1-3]
	\arrow["{{\gamma_x}}"', hook, from=1-1, to=2-1]
	\arrow["{{\gamma_y}}", from=1-3, to=2-3]
	\arrow["{{\varepsilon_P}}"', from=2-1, to=3-2]
	\arrow["{{\varepsilon_P}}", from=2-3, to=3-2]
\end{tikzcd}\]
and in the following proposition:
\begin{proposition} 
\label{prop:sem-symmetry}
For a Petri net $P$, the collection
\[
\mathbb{S}_{\calE(\calU_P)} = \{\theta : x \cong y \mid \theta \in \ES(x, y) \text{\ and\ } \alpha_y \circ \calN(\theta) = \alpha_x\} \] is an isomorphism family on the event structure $\calE(\calU_{P})$ with respect to which the following holds: for every $E \in \ES$ and for all pairs of maps $f_1, f_2 \in \SES(E, \calE(\calU_P))$  such that $f_1 \sim f_2$, $\varepsilon_P \circ \calN(f_1) =\varepsilon_P \circ \calN(f_2)$.
\end{proposition}

For $P \in \Petri$, let $\semSA{P} \in \SES$ denote the pair $(\calE(\calU_{P}), \mathbb{S}_{\calE(\calU_P)})$.

\begin{proof}
    We only detail the extension axiom for isomorphism families. It suffices to look at a one-event extension $x \subseteq x'$ with $x' = x \cup \{ e \}$ and $e \notin x$. The general axiom then follows by induction. Let $e_1, \ldots, e_n$ be the immediate predecessors of $e$ w.r.t.~the order $\leq$. First note that $\gamma_x(e_i) = \gamma_{x'}(e_i)$ for every $i \leq n$, a consequence of naturality.

		Recall the construction of the unfolding: for every $i \leq n$, $\gamma_x(e_i)$ is a transition of $\calU_P$ and so of the form $(A_i, a_i)$. Observe that we must have $\gamma_y(\theta(e_i)) = (A_i', a_i)$ for some $A_i'$, since by assumption both $\gamma_x(e_i)$ and $\gamma_y(\theta(e_i))$ map to the same transition of $P$ (namely, $a_i$). 

	Now let $(B, b)$ denote the transition $\gamma_{x'}(e)$, so (by construction) $B$ is its pre-set of places in $\calU_P$. We partition $B$ as follows: let $B_{\mathrm{int}} \subseteq B$ be the subset of places sitting in between $\gamma_x(e_i)$ and $\gamma_{x'}(e)$ for some $i$ (formally,
       $B_{\mathrm{int}} = \{ s \in B \mid \Post((A_i, a_i), s) \text{ for some } i \leq n \})$ and  
   let $B_{\mathrm{inn}} = B \setminus B_{\mathrm{int}}$. We note that every place $s \in B_{\mathrm{inn}}$ is initial: indeed, any transition of $\calU_P$ having $s$ in its post-set would correspond to an event of $x$ immediately below $e$, thus one of the $e_i$, and we would have $s \in B_{\mathrm{int}}$.
		
   We claim that $(B', b)$ is a well-defined transition of $\calU_P$, where the set $B'$ is defined as $B_{\mathrm{inn}} \cup \{   ((A_i', a_i), s, j) \mid ((A_i, a_i), s, j) \in B_{\mathrm{int}} \}$. To see why this holds, note that every element of $B'$ is a well-defined place: $((A_i', a_i), s, j)$ is valid just because $((A_i, a_i), s, j)$ is, since $\eta(A_i', a_i)= \eta(A_i, a_i)$. Then observe that $\mathrm{co}~B'$ (because the transitions $(A_i', a_i)$ are concurrent and initial places are never involved in a conflict) and the arity condition transfers easily from $B$. So the claim holds and we can let $y' = y \cup \{ (B, b)\}$ and easily verify that the resulting extension of $\theta$ is in $\mathbb S$. 
\end{proof}

\begin{example} We return to the `problematic' example in \S\ref{example-intro}. Applying the above construction, we obtain the same event structure (with two concurrent events), but this time it comes equipped with an isomorphism family that makes the two events symmetric. This is one way to resolve the universality issue by restoring the equivalence, up to symmetry. 
\end{example}

\def\x{2}

In general, in addition to the existence property  of Lemma~\ref{lemmawinskel}, the unfolding satisfies the uniqueness part of a universal property, up to symmetry:
\begin{lemma}
\label{prop:crucial}
    For $O \in \Occ$ and $P \in \Petri$, if $(\eta, \beta): O \rightarrow P$ is a map of Petri nets and there are two maps $(\eta'_1, \beta'_1), (\eta'_2, \beta'_2):  O \rightarrow \calU_P$ satisfying the commutative diagram of Lemma~\ref{lemmawinskel}, 
then $\eta'_1 \sim \eta'_2$ as maps of event structures with symmetry $\calJ\calE O \to \semSA{P}$.
\end{lemma}
\begin{proof}
By assumption the two maps of ocurrence nets are equal when postcomposed by $\varepsilon_P$. The result then follows from our definition of the symmetry in $\calU_P$. 
\end{proof}

\begin{theorem}
\label{thm:pseudo-adj}
  For $E \in \ES$ and $P \in \Petri$, the function 
  \[ 
  \ES(E, \calE{\calU_P}) \xrightarrow{\cong} \Occ(\calN E, \calU_P) \xrightarrow{\varepsilon_P \circ -} \Petri(\calN E, P) 
  \]
extends to an equivalence of categories $\SES(JE, \semSA{P}) \simeq \Petri(\calN E, P)$, natural in $E$. 
\end{theorem}
\begin{proof}
First recall that the domain category $\SES(JE, \semSA{P})$ is a setoid and the codomain $\Petri(\calN E, P)$ is a set, i.e. a discrete category. The function extends to a functor, i.e. symmetric pairs have the same image, by Proposition~\ref{prop:sem-symmetry}. To establish the equivalence it suffices to prove that the function is surjective, which follows from Lemma~\ref{lemmawinskel}. 
\end{proof}

This theorem characterizes the natural transformation $\SES(J-, \semSA{P})$ up to equivalence. Modulo the $J$, this resembles a characterization of $\semSA{P}$ itself by a Yoneda-style argument. We can make this argument precise using a density theorem for $J$. 

\subsection{A 2-density result for event structures}
\label{subsec:2density}

\newcommand{\bE}{\mathbf{E}}
\newcommand{\bC}{\mathbf{C}}
\newcommand{\bD}{\mathbf{D}}

We prove a density result for the embedding of event structures into event structures with symmetry. This is a 2-categorical (i.e. $\Cat$-enriched) form of density: each event structure with symmetry is a canonical 2-colimit of  ordinary event structures. Another way of stating this is as follows:

\newcommand{\op}{\mathrm{op}}

\begin{definition}[e.g. \cite{BCECT}]
For locally small 2-categories $\bD$ and $\bE$, a $2$-functor $J : \bD \rightarrow \bE$ is \emph{2-dense} if the 2-functor $\tilde{J} : \bE \to [\bD^\op, \Cat] : e \mapsto \bE(J-, e)$ is locally an isomorphism of categories, where $[\bD^\op, \Cat]$ is the 2-category of 2-functors, strict natural transformations, and modifications. 
\end{definition}

\begin{theorem}
The embedding 2-functor $J : \ES \to \SES$ is 2-dense.
\end{theorem}
(We emphasize that 2-density is strictly weaker than density in the 1-categorical sense. In particular the underlying 1-functor $J : \ES \to \SES$ is not dense.)
\begin{proof}
For event structures with symmetry $A$ and $B$, we show that the functor $\tilde{J}_{A ,B} : \SES(A, B) \to \Cat^{\ES^\op}(\SES(J-, A), \SES(J-, B))$ is an isomorphism of categories by constructing an inverse. Write $A = (|A|, \mathbb{S}_A)$ and $B = (|B|, \mathbb{S}_B)$. 

Observe that, for a natural transformation $\alpha : \SES(J-, A) \to \SES(J-, B)$, applying component $\alpha_{|A|}$ to the `identity' function $\id_{|A|} : J|A| \to A$ gives a map of event structures with symmetry $\alpha_{|A|}(\id_{|A|}) : J|A| \to B$. We claim that $\alpha_{|A|}(\id_{|A|})$ preserves the symmetry in $A$. For any bijection $\theta: x \cong y$ in $\mathbb{S}_A$, we have $(\id_{|A|} \circ J\iota_y \circ J\theta) \sim (\id_{|A|} \circ J\iota_x)$, writing $\iota_x : x \to |A|$ and $\iota_y : y \to |A|$ for the inclusion maps (and recalling that $\theta \in \ES(x, y)$). Since $\alpha$ is a natural transformation and each component preserves the symmetry relation on maps,  $\alpha_{|A|}(\id_{|A|}) \circ J{\iota_x} = \alpha_x (\id_{|A|} \circ J{\iota_x}) \sim \alpha_x(\id_{|A|} \circ J{\iota_y} \circ J\theta) = \alpha_{|A|}(\id_{|A|}) \circ J\iota_y \circ J\theta.$ In other words, $\alpha_{|A|}(\id_{|A|}) \theta \in \mathbb{S}_B$, and we have proven the claim. Therefore, we take $\tilde J_{A, B}^{-1}(\alpha) = \alpha_{|A|}(\id_{|A|})$. For functoriality of $J_{A, B}^{-1}$, note that any modification $\alpha \to \beta$ must be unique and gives $\alpha_{|A|}(\id_{|A|})\sim \beta_{|A|}(\id_{|A|})$. We omit the verification that $\tilde J_{A, B}^{-1}$ and $\tilde J_{A, B}$ are inverses. 
\end{proof}

From this we extend the unfolding semantics to maps of Petri nets, up to symmetry. 
\begin{corollary}
\label{def:sem}
The unfolding semantics $P \mapsto \semSA{P}$ determines a pseudo-functor $\semAlone : \Petri \to \SES$     (i.e. composition and identities are only preserved up to symmetry). 
\end{corollary}
\begin{proof}
Every map $(\eta, \beta) : P \to P'$ induces a natural transformation $\SES(J -, \semSA{P}) \xrightarrow{\simeq} \Petri(\calN -, P) \xrightarrow{ (\eta, \beta)\circ - } \Petri(\calN -, P') \xrightarrow{\simeq} \SES(J-, \semSA{P'})$, which corresponds to a map $\semSA{P} \to \semSA{P'}$  under the density isomorphism. The non-strictness is because pseudo-inverses are only determined up to isomorphism. 
\end{proof}
It also follows from this construction that the equivalence of categories $\Petri(\calN E, P) \simeq \SES(JE, \semSA{P})$ is natural in $P$. In other words,  we have constructed a \emph{$J$-relative pseudo-adjunction}: 
 \begin{equation}
 \label{eq:relative-adjunction}
 \begin{tikzcd}
	& \SES \\
	\Petri & \ES
	\arrow[""{name=0, anchor=center, inner sep=0}, "\semAlone", from=2-1, to=1-2]
	\arrow["J"', from=2-2, to=1-2]
	\arrow[""{name=1, anchor=center, inner sep=0}, "\calN", from=2-2, to=2-1]
	\arrow["\dashv"{anchor=center, rotate=96}, draw=none, from=1, to=0]
\end{tikzcd}
\end{equation}
The 2-functor $\semAlone$ is determined up to symmetry by the relative 2-adjunction. The general result is as follows. 
\begin{lemma}
  Let $L : \bD \to \bC$ and $\calJ : \bD \to \bE$ be 2-functors and suppose that $\calJ$ is 2-dense. If $L$ has two $\calJ$-relative pseudo-adjoint pseudo-functors $R, R' : \bC \to \bE$, then $R$ and $R'$ are equivalent. 
\end{lemma}
The following diagram summarizes the situation:
\[\begin{tikzcd}
	&& \bE \\
	\bC && \bD
	\arrow[""{name=0, anchor=center, inner sep=0}, "R"{description}, from=2-1, to=1-3]
	\arrow["{R'}", curve={height=-12pt}, from=2-1, to=1-3]
	\arrow["\calJ"', from=2-3, to=1-3]
	\arrow[""{name=1, anchor=center, inner sep=0}, "L", from=2-3, to=2-1]
	\arrow["\dashv"{anchor=center, rotate=-90}, shift left=3, draw=none, from=0, to=1]
\end{tikzcd}\]
\begin{proof}
By the relative pseudo-adjunction property we have for every $c \in \bC$ a natural equivalence of pseudo-functors $\bE(\calJ-, Rc) \simeq \bC(L-, c) \simeq \bE(\calJ-, R'c)$, and this is natural in $c$. By 2-density of $\calJ$ this lifts to an internal equivalence $R c \simeq R'c$ in $\bE$. Several naturality conditions must be verified but this is straightforward diagram chasing. 
\end{proof} 

\subparagraph*{Summary of section.} We have obtained a new presentation of the unfolding semantics for ordinary Petri nets as a pseudo-functor $\semSA{-} : \Petri \to \SES$. The universal property of this unfolding is necessarily weaker than in the whole-grain setting, but we have characterized it as a pseudo-adjunction relative to the embedding $\ES \to \SES$. The 2-density of this embedding suffices to characterize the unfolding up to symmetry.

\section{Multiplicity count as a morphism of relative adjunctions}
\label{sec:adjunction-morphism}

In this section, we connect the unfolding in whole-grain style (Corollary~\ref{cor:whole-grain-unfolding}) and the unfolding in ordinary style (\S\ref{sec:unf-sem-sym}). Observe that there are functors connecting the two settings on all sides of the adjunctions: the multiplicity count functor $\Q : \WGPetri \to \Petri$ and the embedding $\ES \hookrightarrow \SES$. 
The idea is to assemble them into an appropriate \emph{morphism of adjunctions}.

\begin{definition}[Adapted from \cite{RM}]
A \emph{right-morphism of relative adjunctions} between relative adjunctions as on the left below consists of functors $F: \bC \to \bC'$, $G : \bD \to \bD'$ and $H : \bE \to \bE'$, and a pseudonatural transformation $\alpha: HR \to R'F$:
\[
\begin{tikzcd}[row sep=-0.2em, column sep=1em]
	&& \bE \\
	\bC \\
	& {} & \bD
	\arrow[""{name=0, anchor=center, inner sep=0}, "R", from=2-1, to=1-3]
	\arrow["J"', from=3-3, to=1-3]
	\arrow[""{name=1, anchor=center, inner sep=0}, "L", from=3-3, to=2-1]
	\arrow["\dashv"{anchor=center, rotate=90}, draw=none, from=1, to=0]
\end{tikzcd}\quad 
\begin{tikzcd}[row sep=-0.2em, column sep=1em]
	&& \bE' \\
	\bC' \\
	& {} & \bD'
	\arrow[""{name=0, anchor=center, inner sep=0}, "R'", from=2-1, to=1-3]
	\arrow["J'"', from=3-3, to=1-3]
	\arrow[""{name=1, anchor=center, inner sep=0}, "L'", from=3-3, to=2-1]
	\arrow["\dashv"{anchor=center, rotate=90}, draw=none, from=1, to=0]
\end{tikzcd}
\ \  
\begin{tikzcd}[column sep=1.4em]
	\bD & \bE & \bC & \bD \\
	{\bD'} & {\bE'} & {\bC'} & {\bD'}
	\arrow["J", from=1-1, to=1-2]
	\arrow["G"', from=1-1, to=2-1]
	\arrow["H"', from=1-2, to=2-2]
	\arrow["R"', from=1-3, to=1-2]
	\arrow["\alpha", Rightarrow, from=1-2, to=2-3]
    \arrow["F", from=1-3, to=2-3]
	\arrow["L"', from=1-4, to=1-3]
	\arrow["G", from=1-4, to=2-4]
	\arrow["{{J'}}"', from=2-1, to=2-2]
	\arrow["{{R'}}", from=2-3, to=2-2]
	\arrow["{{L'}}", from=2-4, to=2-3]
\end{tikzcd}
\ \ 
\begin{tikzcd}[column sep=1.4em]
	{J'Gd} & {R'L'Gd} &[-1em] {R'FLd} \\
	HJd && HRLd
	\arrow["{\eta'_{Gd}}", from=1-1, to=1-2]
	\arrow[equals, from=1-1, to=2-1]
	\arrow[equals, from=1-2, to=1-3]
	\arrow["{\alpha_{Ld}}", from=1-3, to=2-3]
	\arrow["{H\eta_d}"', from=2-1, to=2-3]
\end{tikzcd}
\]
such that the two unlabelled squares in the middle diagram commute, and that the right-most diagram commutes for every $d \in \bD$ ($\eta$ and $\eta'$ denote the two relative units).

\end{definition}
Here we need a slightly more general notion to deal with the pseudo aspects. Since the 2-categories involved here are quite degenerate, there are no coherence axioms at the 2-cell levels, and we only relax the naturality of $\alpha$:
\begin{theorem}
There is a (pseudo) right-morphism of relative pseudo-adjunctions 
\[\begin{tikzcd}[row sep=-0.2em, column sep=2em]
	& \ES \\
	\WGPetri \\
	& \ES
	\arrow[""{name=0, anchor=center, inner sep=0}, "{\semWG}"{pos=0.8}, from=2-1, to=1-2]
	\arrow[equals, from=3-2, to=1-2]
	\arrow[""{name=1, anchor=center, inner sep=0}, "\N"{pos=0.3}, from=3-2, to=2-1]
	\arrow["\dashv"{anchor=center, rotate=90}, draw=none, from=1, to=0]
\end{tikzcd}
\qquad 
\raisebox{-0.3em}{\scalebox{1.2}{$\longrightarrow$}}
\qquad 
\begin{tikzcd}[row sep=-0em, column sep=2em]
	& \SES \\
	\Petri \\
	& \ES
	\arrow[""{name=0, anchor=center, inner sep=0}, "\semS"{pos=0.8}, from=2-1, to=1-2]
	\arrow["J"', hook', from=3-2, to=1-2]
	\arrow[""{name=1, anchor=center, inner sep=0}, "\calN"{pos=0.3}, from=3-2, to=2-1]
	\arrow["\dashv"{anchor=center, rotate=91}, draw=none, from=1, to=0]
\end{tikzcd}
\]
consisting of the 2-functors $\id_{\ES}: \ES \to \ES$, $\Q : \WGPetri \to \Petri$, and $J : \ES \to \SES$.
\end{theorem}
\begin{proof}
We define a pseudo-natural transformation $\alpha_{P} : J \semWGA{P} \to \semSA{\Q(P)}$ for $P \in \WGPetri$ and omit the rest of the details.  For $P \in \WGPetri$, $\alpha_P$ is built from the counit $\varepsilon_P$ of the adjunction between $\WGPetri$ and $\Occ$. Its image $\Q\varepsilon_P : \Q\U P \rightarrow \Q P$ factors through a morphism $g_P : \Q P \rightarrow \calU_{\Q P}$   by Lemma \ref{lemmawinskel}, and we define $\alpha_P$ as $\calE g_P$ seen as a morphism $J\semWGA{P} \to \semSA{\Q P}$. Pseudonaturality follows from Lemma \ref{prop:crucial}.
\end{proof}

\section{Conclusion: related work and perspectives}

The symmetry problems that arise in the unfolding of unsafe Petri nets are well-known but difficult to explain precisely. We have attempted to give an account of the situation using 2-categorical language.

One key takeaway is that event structures with symmetry are an appropriate semantic domain for Petri nets, whether safe or unsafe. The corresponding unfolding construction is characterized as a right relative 2-adjoint. This method completely bypasses the difficulties of adding symmetry on Petri nets themselves \cite{GENUN}.

Over the past few years the theory of Petri nets has experienced a new wave of interest, with many new contributions on the categorical side 
(\cite{baez2021categories,baez2020open,bumpus2025additive,master2025colored, ACFPN}) and new applications to semantics \cite{di2025dialectica,castellan2023geometry} and practical systems modelling \cite{dalrymple2024safeguarded,catcolab}. We hope that the present work will resonate with this line of work. 

Meanwhile, event structures and symmetry have applications to program semantics (e.g. \cite{GSES,castellan2015parallel}) and the present work may provide new perspectives in that area. We could also explore connections with other kinds of unfoldings where the symmetry problems also occur. For instance \cite{ESES, SIC, GCMC} have all (implicitly or  explicitly) considered symmetry to describe the complexities of unfolding semantics. 

\bibliographystyle{alpha} 
\bibliography{biblio}

\appendix
\section{The bicategory $\WGPetri$ is locally a setoid}
\label{sec:discreteness-proof}

This appendix gives a detailed proof that there can be at most one 2-cell between two rational maps of whole-grain Petri nets. First we recall the theorem. Note that this property holds only because we have made the convenient assumption that Petri nets have no isolated places, an essential assumption for the proof to go through.  

\discreteness*

For the proof, let the four maps involved be typed as
\[
P \xleftarrow{\varphi} Q \xrightarrow{\psi} P'  \qquad P \xleftarrow{\varphi'} Q' \xrightarrow{\psi'} P'
\]
where $Q$ and $Q'$ are whole-grain Petri nets, $\varphi$ and $\varphi'$ are cabling maps, and $\psi$ and $\psi'$ are \'etale maps. Following our convention, let the data of the four Petri nets be denoted by 
\begin{align*}
 & P\quad =\quad  S_P \leftarrow  I_P \to T_P \leftarrow  O_P \to S_P \\
& P'\quad =\quad  S_{P'} \leftarrow  I_{P'} \to T_{P'} \leftarrow  O_{P'} \to S_{P'} \\
& Q\quad =\quad  S_Q \leftarrow  I_Q \to T_Q \leftarrow  O_Q \to S_Q \\
& Q'\quad =\quad  S_{Q'} \leftarrow  I_{Q'} \to T_{Q'} \leftarrow  O_{Q'} \to S_{Q'}
\end{align*}
so that, in particular, there are commutative diagrams and pullbacks as follows:
\[\begin{tikzcd}
	{S_P} & {I_P} & {T_P} & {O_P} & {S_P} \\
	{S_Q} & {I_Q} & {T_Q} & {O_Q} & {S_Q} \\
	{S_{P'}} & {I_{P'}} & {T_{P'}} & {O_{P'}} & {S_{P'}}
	\arrow[from=1-2, to=1-1]
	\arrow[from=1-2, to=1-3]
	\arrow[from=1-4, to=1-3]
	\arrow[from=1-4, to=1-5]
	\arrow[from=2-1, to=1-1]
	\arrow[from=2-1, to=3-1]
	\arrow["\lrcorner"{anchor=center, pos=0.125, rotate=180}, draw=none, from=2-2, to=1-1]
	\arrow[from=2-2, to=1-2]
	\arrow[from=2-2, to=2-1]
	\arrow[from=2-2, to=2-3]
	\arrow[from=2-2, to=3-2]
	\arrow["\lrcorner"{anchor=center, pos=0.125}, draw=none, from=2-2, to=3-3]
	\arrow[equals, from=2-3, to=1-3]
	\arrow[from=2-3, to=3-3]
	\arrow[from=2-4, to=1-4]
	\arrow["\lrcorner"{anchor=center, pos=0.125, rotate=90}, draw=none, from=2-4, to=1-5]
	\arrow[from=2-4, to=2-3]
	\arrow[from=2-4, to=2-5]
	\arrow["\lrcorner"{anchor=center, pos=0.125, rotate=-90}, draw=none, from=2-4, to=3-3]
	\arrow[from=2-4, to=3-4]
	\arrow[from=2-5, to=1-5]
	\arrow[from=2-5, to=3-5]
	\arrow[from=3-2, to=3-1]
	\arrow[from=3-2, to=3-3]
	\arrow[from=3-4, to=3-3]
	\arrow[from=3-4, to=3-5]
\end{tikzcd}\]
\[\begin{tikzcd}
	{S_P} & {I_P} & {T_P} & {O_P} & {S_P} \\
	{S_{Q'}} & {I_{Q'}} & {T_{Q'}} & {O_{Q'}} & {S_{Q'}} \\
	{S_{P'}} & {I_{P'}} & {T_{P'}} & {O_{P'}} & {S_{P'}}
	\arrow[from=1-2, to=1-1]
	\arrow[from=1-2, to=1-3]
	\arrow[from=1-4, to=1-3]
	\arrow[from=1-4, to=1-5]
	\arrow[from=2-1, to=1-1]
	\arrow[from=2-1, to=3-1]
	\arrow["\lrcorner"{anchor=center, pos=0.125, rotate=180}, draw=none, from=2-2, to=1-1]
	\arrow[from=2-2, to=1-2]
	\arrow[from=2-2, to=2-1]
	\arrow[from=2-2, to=2-3]
	\arrow[from=2-2, to=3-2]
	\arrow["\lrcorner"{anchor=center, pos=0.125}, draw=none, from=2-2, to=3-3]
	\arrow[equals, from=2-3, to=1-3]
	\arrow[from=2-3, to=3-3]
	\arrow[from=2-4, to=1-4]
	\arrow["\lrcorner"{anchor=center, pos=0.125, rotate=90}, draw=none, from=2-4, to=1-5]
	\arrow[from=2-4, to=2-3]
	\arrow[from=2-4, to=2-5]
	\arrow["\lrcorner"{anchor=center, pos=0.125, rotate=-90}, draw=none, from=2-4, to=3-3]
	\arrow[from=2-4, to=3-4]
	\arrow[from=2-5, to=1-5]
	\arrow[from=2-5, to=3-5]
	\arrow[from=3-2, to=3-1]
	\arrow[from=3-2, to=3-3]
	\arrow[from=3-4, to=3-3]
	\arrow[from=3-4, to=3-5]
\end{tikzcd}\]

A 2-cell $(\varphi, \psi) \to (\varphi', \psi')$ is a map of whole-grain Petri nets $Q \to Q'$ such that the triangles
\[\begin{tikzcd}
	& Q \\
	P && {P'} \\
	& {Q'}
	\arrow[from=1-2, to=2-1]
	\arrow[from=1-2, to=2-3]
	\arrow["\cong"{description}, from=1-2, to=3-2]
	\arrow[from=3-2, to=2-1]
	\arrow[from=3-2, to=2-3]
\end{tikzcd}\]
commute. Suppose that we have two 2-cells $\alpha, \beta : Q \to Q'$ respectively given by the families of maps below:
\[\begin{tikzcd}
	{S_Q} & {I_Q} & {T_Q} & {O_Q} & {S_Q} \\
	{S_{Q'}} & {I_{Q'}} & {T_{Q'}} & {O_{Q'}} & {S_{Q'}}
	\arrow["\alpha_S"', from=1-1, to=2-1]
	\arrow[from=1-2, to=1-1]
	\arrow[from=1-2, to=1-3]
	\arrow["\alpha_I"', from=1-2, to=2-2]
	\arrow["\alpha_T"', from=1-3, to=2-3]
	\arrow[from=1-4, to=1-3]
	\arrow[from=1-4, to=1-5]
	\arrow["\alpha_O"', from=1-4, to=2-4]
	\arrow["\alpha_S", from=1-5, to=2-5]
	\arrow[from=2-2, to=2-1]
	\arrow[from=2-2, to=2-3]
	\arrow[from=2-4, to=2-3]
	\arrow[from=2-4, to=2-5]
\end{tikzcd}\]

\[\begin{tikzcd}
	{S_Q} & {I_Q} & {T_Q} & {O_Q} & {S_Q} \\
	{S_{Q'}} & {I_{Q'}} & {T_{Q'}} & {O_{Q'}} & {S_{Q'}}
	\arrow["\beta_S"', from=1-1, to=2-1]
	\arrow[from=1-2, to=1-1]
	\arrow[from=1-2, to=1-3]
	\arrow["\beta_I"', from=1-2, to=2-2]
	\arrow["\beta_T"', from=1-3, to=2-3]
	\arrow[from=1-4, to=1-3]
	\arrow[from=1-4, to=1-5]
	\arrow["\beta_O"', from=1-4, to=2-4]
	\arrow["\beta_S", from=1-5, to=2-5]
	\arrow[from=2-2, to=2-1]
	\arrow[from=2-2, to=2-3]
	\arrow[from=2-4, to=2-3]
	\arrow[from=2-4, to=2-5]
\end{tikzcd}\]
We show that the vertical maps are pairwise equal, so that $\alpha = \beta$.

($\alpha_T = \beta_T$). This is because $T_Q = T_{Q'} = T_P$ and the diagram
\[\begin{tikzcd}
	& {T_Q} \\
	{T_P} \\
	& {T_{Q'}}
	\arrow[equals, from=1-2, to=2-1]
	\arrow["\alpha_T", from=1-2, to=3-2]
	\arrow[equals, from=3-2, to=2-1]
\end{tikzcd}\]
must commute, so $\alpha_T$ is an identity map. The same argument also applies to $\beta_T$.

($\alpha_O = \beta_O$).
The diagram below commutes whether the map $O_Q \to O_{Q'}$ is $\alpha_O$ or $\beta_O$.
\[\begin{tikzcd}
	{T_Q} & {O_Q} \\
	{T_{Q'}} & {O_{Q'}} \\
	{T_{P'}} & {O_{P'}}
	\arrow[equals, from=1-1, to=2-1]
	\arrow[curve={height=30pt}, from=1-1, to=3-1]
	\arrow[from=1-2, to=1-1]
	\arrow["\alpha_O \text{ or } \beta_O"', from=1-2, to=2-2]
	\arrow[curve={height=-30pt}, from=1-2, to=3-2]
	\arrow[from=2-1, to=3-1]
	\arrow[from=2-2, to=2-1]
	\arrow["\lrcorner"{anchor=center, pos=0.125, rotate=-90}, draw=none, from=2-2, to=3-1]
	\arrow[from=2-2, to=3-2]
	\arrow[from=3-2, to=3-1]
\end{tikzcd}\]
Since the bottom square is a pullback, the map $O_Q \to O_{Q'}$ is unique, so $\alpha_O = \beta_O$. Since the outer square is a pullback too ($\psi$ is \'etale), so is the top square. Isomorphisms are stable under pullback, therefore $\alpha_O$ is invertible. 

($\alpha_I = \beta_I$). Symmetric argument.

($\alpha_S = \beta_S$). The squares
\[\begin{tikzcd}
	{S_{Q'}} & {I_{Q'}} && {S_{Q'}} & {O_{Q'}} \\
	{S_P} & {I_P} && {S_P} & {O_P}
	\arrow[from=1-1, to=2-1]
	\arrow[from=1-2, to=1-1]
	\arrow["\lrcorner"{anchor=center, pos=0.125, rotate=-90}, draw=none, from=1-2, to=2-1]
	\arrow[from=1-2, to=2-2]
	\arrow[from=1-4, to=2-4]
	\arrow[from=1-5, to=1-4]
	\arrow["\lrcorner"{anchor=center, pos=0.125, rotate=-90}, draw=none, from=1-5, to=2-4]
	\arrow[from=1-5, to=2-5]
	\arrow[from=2-2, to=2-1]
	\arrow[from=2-5, to=2-4]
\end{tikzcd}\]
are pullbacks and since $\mathbf{Set}$ is an extensive category, the square
\[\begin{tikzcd}
	{S_{Q'}} & {I_{Q'} + O_{Q'}} \\
	{S_P} & {I_P + O_P}
	\arrow[from=1-1, to=2-1]
	\arrow[from=1-2, to=1-1]
	\arrow["\lrcorner"{anchor=center, pos=0.125, rotate=-90}, draw=none, from=1-2, to=2-1]
	\arrow[from=1-2, to=2-2]
	\arrow[from=2-2, to=2-1]
\end{tikzcd}\]
whose horizontal maps are copairings is also a pullback. The same argument shows that
\[\begin{tikzcd}
	{S_Q} & {I_Q + O_Q} \\
	{S_P} & {I_P + O_P}
	\arrow[from=1-1, to=2-1]
	\arrow[from=1-2, to=1-1]
	\arrow["\lrcorner"{anchor=center, pos=0.125, rotate=-90}, draw=none, from=1-2, to=2-1]
	\arrow[from=1-2, to=2-2]
	\arrow[from=2-2, to=2-1]
\end{tikzcd}\]
is also a pullback. We have assumed that $I_P + O_P \to S_P$ is always surjective for a whole-grain Petri net (no isolated places), and pullbacks of surjective maps are surjective in $\mathbf{Set}$, so $I_Q + O_Q \to S_Q$ is surjective. Since the following diagram commutes for either $\alpha_S$ or $\beta_S$, 
\[\begin{tikzcd}
	{S_Q} & {I_Q + O_Q} \\
	{S_{Q'}} & {I_{Q'} + O_{Q'}}
    \arrow["\alpha_S \text{ or } \beta_S"', from=1-1, to=2-1]
	\arrow[from=1-2, to=1-1]
	\arrow[from=1-2, to=2-2]
	\arrow[from=2-2, to=2-1]
\end{tikzcd}\]
the defining property of epimorphisms implies that $\alpha_S = \beta_S$. It is in fact a pullback square (apply the pullback pasting law as above) and we have shown that the right-hand vertical arrow $\alpha_I + \alpha_O$ is invertible, therefore $\alpha_S$ is invertible too. 

Altogether we have shown that there is at most one 2-cell $(\varphi, \psi) \to (\varphi', \psi')$ and that when it exists it is invertible.

\end{document}